\documentclass[11pt]{article}

\usepackage[margin=1in]{geometry}
\usepackage{amsmath,amssymb,amsthm,mathtools}
\usepackage{booktabs,threeparttable,longtable,tabularx,array}
\usepackage{graphicx}
\usepackage{subcaption}
\usepackage{natbib}
\usepackage{setspace}
\usepackage{microtype}
\usepackage{hyperref}
\usepackage{enumitem}
\usepackage{float}
\usepackage{pdflscape}
\usepackage[most]{tcolorbox}
\usepackage{authblk}

\hypersetup{
	colorlinks=true,
	citecolor=blue,
	linkcolor=blue,
	urlcolor=blue,
	hypertexnames=false,
	pdftitle={A Multiscale Ball Test for Conditional Mean Independence},
	pdfauthor={Simon Rudkin and Wanling Rudkin},
	pdfsubject={Theory, Monte Carlo evidence, and finance applications for a multiscale conditional-mean test},
	pdfkeywords={mean independence, multiscale testing, local dependence, wild bootstrap, prewhitening, conditional heteroskedasticity}
}

\newtheorem{theorem}{Theorem}
\newtheorem{proposition}{Proposition}
\newtheorem{lemma}{Lemma}
\newtheorem{corollary}{Corollary}
\newtheorem{assumption}{Assumption}
\newtheorem{remark}{Remark}

\newcommand{\E}{\mathbb{E}}
\newcommand{\Pp}{\mathbb{P}}
\newcommand{\R}{\mathbb{R}}

\newcommand{\1}{\mathbf{1}}

\newcolumntype{Y}{>{\raggedright\arraybackslash}X}

\author[1]{Simon Rudkin \thanks{\textbf{Corresponding Author} Full Address: Department of Social Statistics, School of Social Sciences, University of Manchester, Oxford Road, Manchester, M13 9PL, United Kingdom. Tel: +44 (0)7955 109334 Email:simon.rudkin@manchester.ac.uk}}
\affil[1]{School of Social Sciences, University of Manchester, United Kingdom}

\vspace{-30pt}
\author[2]{Wanling Rudkin \thanks{Full Address: University of Exeter Business School, University of Exeter, Streatham Court, Rennes Drive, Exeter, EX4 4PU, United Kingdom.  Email: w.rudkin@exeter.ac.uk}}
\affil[2]{University of Exeter Business School, University of Exeter, United Kingdom}
\vspace{-40pt}
\vspace{-20pt}

\title{A Multiscale Ball Test for Conditional Mean Independence}

\begin{document}

\maketitle

\begin{abstract}

Tests of conditional mean independence can lose power when departures are confined to a bounded part of a multivariate predictor space and the relevant spatial scale is unknown. We propose a Multiscale Ball Conditional Mean Independence (MBCMI) test that aggregates support-weighted local mean contrasts in an outcome variable across balls centered on each data point in a predictor set. Fixed-grid theory identifies the population target, establishes consistency for grid-visible alternatives, and derives a Pitman local-power limit governed by the ball-smoothed mean departure. For serial data, feasible recursive-sign-bootstrap validity for stable finite-order autoregressions with conditionally sign-symmetric innovations is established. Application-aligned serial null experiments reject 4.25\% of the time. MBCI is demonstrated to be strongest for local and radial signals. Predictor-law experiments show that these conclusions are not an artefact of independent Gaussian covariates. In monthly U.S. finance data, cross-fitted residual MBCMI tests remove every full-sample rejection, suggesting contemporaneous conditional-mean dependence rather than evidence of distinctive nonlinear structure. 
\end{abstract}


\section{Introduction}

Mean independence of a scalar outcome $Y$ from a multivariate predictor $X$ requires
\begin{equation}
H_0:\quad \E[Y\mid X]=\E[Y]\quad\text{almost surely}.
\label{eq:null}
\end{equation}
The restriction is weaker than full independence but economically relevant whenever the object of interest is a conditional mean. Global conditional-moment procedures integrate discrepancies over the predictor support \citep{bierens1990conditional,bierens1997integrated,fanli1996consistent,zheng1996functional}. That aggregation can dilute a departure confined to a bounded region. Local smoothers avoid global averaging but create a tuning problem: the spatial scale at which signal separates from noise is rarely known in advance.

We develop a Multiscale Ball Conditional Mean Independence (MBCMI) test as a response to the need to balance local smoothing and robustness to noise. Every observed predictor vector in $X$ is a centre around which a ball of identical radius is drawn. At each radius, the statistic compares the outcome mean inside every supported Euclidean ball with the global outcome mean and weights the squared contrast by neighbourhood size. The test maximises the resulting statistic over empirical pairwise-distance quantiles $q=.05,.06,\ldots,.75$ and repeats that complete radius search inside every resampling draw. In total there are 71 radii considered in each application of the MBCMI test. The same construction yields a centre-by-radius decomposition that describes where the statistic accumulates. Selected radius and contribution regions are diagnostic objects, not structural parameters or causal regimes.

Ball geometry matters because it fixes a common spatial extent while allowing the probability mass inside that extent to vary across the predictor support. That architecture differs from choosing one smoothing bandwidth, integrating a discrepancy globally, or fixing a neighbour count. Observed centres avoid an external mesh, support weighting aligns the sample statistic with a population local-mean criterion, and joint calibration of the radius maximum accounts for the scale search. MBCMI's contribution therefore lies in the combination of a fixed-extent multiscale test and an interpretable centre-by-scale decomposition, not in metric balls alone.

The contribution of MBCMI is best understood relative to estimand-matched procedures. Integrated conditional-moment and kernel tests share the conditional-restriction objective but use different aggregations \citep{bierens1990conditional,fanli1996consistent,zheng1996functional,suwhite2007,suwhite2014}. Martingale-difference correlation and divergence use global distance functionals for conditional-mean constancy \citep{shaozhang2014mdc,leeshao2018mddm}. The recent nearest-neighbour graph conditional-mean procedure of \citet{chatterjee2026cmi} provides an especially close benchmark. We evaluate its author-specified $K=5$ and $K=10$ versions on the exact canonical datasets. MBCMI fixes spatial extent and lets local mass vary; nearest-neighbour procedures fix local mass and let spatial extent vary. That distinction predicts the empirical power pattern: nearest-neighbour graph methods excel for rapidly sign-changing interaction and checkerboard alternatives, while the fixed-radius ball neighbourhoods of MBCMI perform well for local islands, annular departures, and local absolute signals.

Full-independence methods remain useful geometric comparators rather than estimand-matched substitutes. Distance correlation aggregates pairwise distances \citep{szekely2007distance}; HSIC compares kernel mean embeddings \citep{gretton2005kernel,gretton2010consistent}; KSG estimates mutual information through nearest-neighbour counts \citep{kraskov2004estimating}; and MGC searches graph scales \citep{shen2020mgc}. The proposed statistic is also distinct from Ball covariance \citep{pan2020ballcov} and Ball Mapper \citep{dlotko2019ballmapper}: those methods share metric-ball primitives but not the population criterion or calibration.

\begin{table}[!htbp]
\centering
\caption{Relation to the closest testing architectures}
\label{tab:literature-positioning}
\begin{threeparttable}
\small
\begin{tabularx}{\textwidth}{@{}p{0.22\textwidth}p{0.19\textwidth}p{0.22\textwidth}X@{}}
\toprule
Procedure & Population target & Geometry or aggregation & Relation to multiscale Ball \\
\midrule
Integrated conditional-moment and kernel specification tests & Conditional restrictions or functional form & Global integration over transforms or bandwidth-specific smooths & Share a conditional-moment objective but not the same centre-by-radius decomposition. \\
MDC/MDD & Conditional-mean constancy & Global distance functional & Important estimand-matched family. The public author-code MDD implementation was numerically feasible but did not satisfy the prespecified finite-sample calibration criterion used for the canonical power tournament. \\
NCMD $K=5,10$ & Conditional-mean constancy & Nearest-neighbour graph statistic & Evaluated on exact paired canonical datasets. Especially strong for interaction and checkerboard alternatives; Ball is stronger for several fixed-spatial-extent local alternatives. \\
$k$NN local means & Conditional-mean constancy & Fixed neighbour count with adaptive geometric extent & Ball fixes spatial extent and lets local mass vary; $k$NN fixes mass and lets extent vary. \\
dCor, HSIC, KSG, and MGC & Full independence & Global distances, kernels, mutual information, or multiscale graphs & Useful geometric benchmarks but can respond to distributional features beyond the conditional mean. \\
Ball covariance & Full independence & Probabilities of metric balls & Shares ball geometry but defines a different population functional. \\
Ball Mapper & Exploratory data-shape summary & Metric-ball cover graph & Motivates radius-cover interpretation but is neither the statistic nor its calibration. \\
\bottomrule
\end{tabularx}
\begin{tablenotes}[flushleft]\footnotesize
\item Notes: NCMD refers to the nearest-neighbour conditional-mean procedure of \citet{chatterjee2026cmi}. NCMD $K=5$ and $K=10$ are reported separately; no data-driven choice of $K$ is made.
\end{tablenotes}
\end{threeparttable}
\end{table}

Inference from conditional independence hinges on the sampling environment. Unrestricted outcome permutation is exact under exchangeability but can fail under heteroskedastic mean independence. The iid procedure therefore uses Rademacher outcome multipliers, preserving observation-specific residual magnitudes \citep{wu1986jackknife,mammen1993bootstrap}. Serial applications use a BIC-selected prewhitened recursive Rademacher bootstrap: fitted autoregressive persistence is removed, dated innovations receive random signs, and the fitted dynamic law recursively recolours the signed innovations. Related specification-testing work shows why bootstrap validity must track the nuisance estimation and dependence structure used in implementation \citep{lavalocal2019,fosten2020nowcast,cavaliere2024garch,fosten2024quantile}. MBCMI is demonstrated robust to alternative handling of the sampling environment, including circular shifts and block procedures.

An important distinction is drawn between formal theory and computational validation. Theory and empirical validation answer different questions. Fixed-grid theory identifies the population criterion and establishes consistency only for alternatives visible at an indexed radius; a new Pitman result shows that local power depends on the Ball-smoothed departure at those same radii. For iid observations, a Ball-specific kernel reduction turns the implemented raw radius vector into a finite collection of degenerate V-statistics, and a multiplier argument gives conditional Rademacher validity for the raw 71-radius maximum under primitive support, moment, local-mass, and distance-density conditions. Serial inference uses a separate recursive-sign construction. Under a stable finite-order autoregression with conditionally sign-symmetric innovations, the oracle sign randomisation is conditionally exact and BIC/least-squares replacement is asymptotically negligible, which yields a feasible validity result for the implemented finite-grid maximum. Conditional sign symmetry is stronger than the headline mean-independence null. Broader asymmetric martingale-difference dynamics therefore remain simulation-validated rather than theorem-covered.

Calibration and power evidence sharpen the method in four ways. First, endpoint diagnostics show that a $q\le .50$ search can bind in application-sized samples. A dense extension to $q=.75$ passes application-aligned calibration: the raw serial maximum rejects 4.25\% of 16,000 null replications, with cell rates between 2.6\% and 5.8\%. A tail sensitivity exercise with $q=.80,.85,.90$ confirms that MBCMI approach continues to perform as expected when each ball covers a large part of the dataset. Second, studentising radius-specific statistics raises average canonical power by only about half a percentage point and is mildly conservative in the application-aligned serial null; the raw maximum remains primary. Third, predictor-law experiments show that ordinary z-scoring is not an artefact of independent Gaussian designs. Robust-MAD scaling performs similarly in most laws, while whitening can help clustered support but materially hurt correlated Gaussian and $t$ designs. Fourth, the comparator evidence reveals a geometry-specific power envelope rather than universal dominance. NCMD $K=10$ has the largest average full-range AUC, but Ball is strongest for several local fixed-extent alternatives and is competitive for weak linear signals.

The MBCMI testing approach is demonstrated for two finance applications. In both cases, monthly U.S. data from January 1980 through September 2025 is used. First, the factor spanning question for the \cite{fama2015five} 5 factors plus momentum factor \citep{carhart1997persistence} is explored by asking which factor(s) returns are independent from the joint distribution of the returns on the other 5 factors. Six factor returns rotate through the outcome position, with the remaining factors as contemporaneous predictors. For the second application, four macro-finance equations study market excess returns, industrial-production growth, unemployment-rate change, and Baa-spread change. Under the final $q\le .75$ rule, five factor equations reject after within-family Holm adjustment and RMW does not. Unemployment-rate change is the only macro rejection after Holm adjustment, although its family-wise status is sensitive to the $q\le .90$ tail extension. These full-sample rejections largely overlap with ordinary contemporaneous linear structure: a nuisance-refitted cross-fitted residual Ball procedure rejects none of the ten full-sample equations. Among five selected diagnostic windows, only Momentum in July 2022 remains significant after the cross-fitted linear benchmark. Rejecting conditional independence of a factor, or macro-financial variable does not mean that there is a guaranteed relationship between 

The remainder of the paper is organised as follows. Section~\ref{sec:method} defines the MBCMI statistic and comparators. Section~\ref{sec:inference} states the inferential framework and calibration. Section~\ref{sec:mc} reports canonical and robustness evidence for the MBCMI test. Section~\ref{sec:finance} presents two applications of MBCMI to financial time series data. Section~\ref{sec:discussion} discusses interpretation and limitations.

\section{Multiscale Ball Test}
\label{sec:method}

Let $(Y_i,X_i)$, $i=1,\ldots,n$, denote observations with
$Y_i\in\R$ and $X_i\in\R^d$. Write
\[
m(x)=\E[Y\mid X=x],
\qquad
\mu=\E[Y].
\]
MBCMI tests mean independence of $Y$ from $X$,
\begin{equation}
H_0:\quad m(X)=\mu\quad\text{almost surely}.
\label{eq:method-null}
\end{equation}
Full independence implies \eqref{eq:method-null}, but the converse fails when
conditional variances, tails, or higher moments depend on $X$. Accordingly,
the statistic targets one feature of the conditional distribution rather than an
omnibus independence restriction.

\subsection{Population Criterion}

Let $a=\E[X]$, let $D$ contain the positive marginal standard deviations of
$X$, and define the population-standardised predictor $Z=D^{-1}(X-a)$. For
an independent copy $(Y',Z')$, set
\begin{align}
p_e(z)&=\Pp\{\|Z'-z\|\le e\},\\
r_e(z)&=\E[Y'\1\{\|Z'-z\|\le e\}],\\
m_e(z)&=r_e(z)/p_e(z),
\end{align}
whenever $p_e(z)>0$. Its population criterion is
\begin{equation}
Q(e)=\E\!\left[p_e(Z)\{m_e(Z)-\mu\}^2\right].
\label{eq:population-criterion}
\end{equation}
Criterion \eqref{eq:population-criterion} averages squared local mean
contrasts and weights each contrast by the probability mass that supports it.
Weighting serves two purposes. It prevents a small, weakly supported
neighbourhood from counting as much as a broad predictor region, and it
matches the sample statistic's neighbourhood-size weight.

Null \eqref{eq:method-null} implies $Q(e)=0$ at every radius. Its converse is
local rather than global. Under compact regular support, positive local mass,
and almost-everywhere continuity of $m$, every nonconstant conditional mean
gives $Q(e)>0$ for all sufficiently small positive radii. Small balls recover
the local regression function before averaging can cancel separated regions.
This identification result explains why sufficiently small population radii
detect nonconstancy. Fixed-grid consistency additionally requires visibility
at one of the fixed radii. Rapidly alternating signals expose the
finite-sample limit: an identifying radius may exist but contain too few
observations for reliable estimation.
Appendix~\ref{app:population} gives the formal statements and proofs on the role of radii.

Observed predictors define ball centres. Using all observed predictors avoids an external
mesh and keeps the statistic invariant to empty parts of a bounding box. Additionally, more centres are placed where the predictor distribution places more mass.
This design adapts to the data without changing the common geometric
radius.

Observed centres generate overlapping balls. Overlap is intentional. A
partition assigns each observation to one cell and can create artificial
boundaries. Overlapping balls let one observation support several nearby
local contrasts. Dependence among centre-level contributions is the cost.
Fixed-grid theory treats that dependence directly through the
quadratic-form representation rather than pretending that centres are
independent replications.

Boundary centres receive asymmetric neighbourhoods when support ends within
radius $e$. Compact-support and local-mass conditions control this support effect in
the formal results. In finite samples, minimum-size and coverage rules remove
centres and radii that cannot support stable local means. Having safeguards on the support does
not repair sparse high-dimensional geometry. Safeguards only prevent unsupported
local estimates from determining the maximum.

\subsection{Sample Statistic}
\label{sec:statistic}

Sample means and marginal standard deviations replace the population values $a$ and $D$. Write the
resulting predictors as $\widehat Z_i$. For centre $i$ and radius $e>0$,
define
\begin{equation}
B_i(e)=\{j:\|\widehat Z_j-\widehat Z_i\|\le e\},
\qquad
N_i(e)=|B_i(e)|,
\end{equation}
and
\[
\bar Y_i(e)=N_i(e)^{-1}\sum_{j\in B_i(e)}Y_j,
\qquad
\bar Y=n^{-1}\sum_{j=1}^nY_j.
\]
For minimum neighbourhood size $N_{\min}$, the fixed-radius statistic is
\begin{equation}
T_n(e)=
\sum_{i\in\mathcal I_n(e)}
N_i(e)\{\bar Y_i(e)-\bar Y\}^2,
\label{eq:T-eps}
\end{equation}
where $\mathcal I_n(e)=\{i:N_i(e)\ge N_{\min}\}$. Each summand is a
centre-level contribution,
\begin{equation}
C_i(e)=N_i(e)\{\bar Y_i(e)-\bar Y\}^2.
\label{eq:center-contribution}
\end{equation}
Large $C_i(e)$ can reflect a large local contrast, broad local support, or
both. This decomposition therefore identifies where the statistic
accumulates, but it does not assign causal meaning to the associated
predictor states.

Neighbourhood-size weighting distinguishes \eqref{eq:T-eps} from an
unweighted average of local test statistics. Suppose two centres have the
same local mean contrast but one ball contains four times as many
observations. A larger ball estimates a feature of a broader part of the
predictor distribution and receives four times the weight. Neighbourhood weighting
aligns the sample statistic with $Q(e)$. It also makes the raw null scale vary
with radius, density, and overlap. Resampling must therefore repeat the whole
radius search rather than compare the observed maximum with fixed-radius
critical values.

\subsection{Choice of Radius}

A radius remains admissible when at least $\gamma n$ centres satisfy the
minimum-size rule:
\begin{equation}
|\mathcal I_n(e)|\ge\lceil\gamma n\rceil.
\label{eq:coverage}
\end{equation}
Every reported analysis uses $\gamma=0.20$ and
$N_{\min}=10$. A paired 20\%--40\% audit changed no selected radius, observed
maximum, resampling $p$-value, or rejection decision in 13,000 Monte Carlo
comparisons or the four selected diagnostic windows of the examples presented in this paper. Coverage screening therefore
removes unsupported low-end candidates without determining the paper's
reported conclusions.

We use the fixed quantile-index set
\[
\mathcal Q=\{0.05,0.06,\ldots,0.75\}.
\]
Let $\widehat e_{n,q}$ denote the empirical $q$th quantile of pairwise
distances among the $\widehat Z_i$. Quantile indexing performs two useful
normalisations. Using quantile indexing rather than specific values adapts numerical radii to the predictor distribution, and
it keeps the number of candidate scales fixed across sample sizes and
dimensions. Quantile indexing does not make geometric scale comparable in a
structural sense. A given quantile can correspond to very different
numerical radii and neighbourhood configurations across designs.

Endpoint diagnostics on a shorter $q\le .50$ candidate set showed substantial upper-bound concentration in application-sized rolling samples. We therefore evaluated the dense extension through $q=.75$ before fixing the primary domain. The extension retains nominal application-aligned size and reduces boundary selection materially. A sparse tail sensitivity adds $q\in\{.80,.85,.90\}$ to assess dependence on still larger scales. The .75 endpoint is therefore a calibrated practical boundary, not a theoretically unique scale.

Multiscale aggregation gives
\begin{equation}
T_n^{\max}
=
\max_{q\in\widehat{\mathcal Q}_n}
T_n(\widehat e_{n,q}),
\label{eq:maxT}
\end{equation}
where $\widehat{\mathcal Q}_n$ contains the admissible indices. The
descriptive selected scale is
\begin{equation}
\widehat e_n^*
=
\arg\max_{q\in\widehat{\mathcal Q}_n}
T_n(\widehat e_{n,q}).
\label{eq:eps-star}
\end{equation}
Every calibration repeats coverage screening, all 71 radius evaluations, and
maximisation. The 71 radii follow from using $q \in \left[0.05,0.75\right]$ with intervals of 0.01. Inference concerns $T_n^{\max}$. Its selected radius explains
which scale generated the observed maximum; it is not a separately tested or
consistently estimated structural parameter.

\paragraph{Quadratic form and diagnostics.}

A quadratic-form representation clarifies overlap and null scaling. Let
$W_e$ contain the admissible local-mean weights, let $D_e$ contain
neighbourhood sizes, and set $A_e=W_e-\mathbf 1\mathbf 1'/n$. Then
\begin{equation}
T_n(e)=Y'A_e'D_eA_eY.
\label{eq:quadratic-form}
\end{equation}
Under the null, conditional covariance $\Omega_X$ gives
\begin{equation}
\E[T_n(e)\mid X]
=\operatorname{tr}(A_e'D_eA_e\Omega_X).
\label{eq:null-mean}
\end{equation}
Neighbourhood overlap, local mass, and heteroskedasticity can therefore alter
the null location and scale across radii. This observation motivates wild
calibration under heteroskedasticity and explains why raw and studentised
maxima can select different scales.

Three profiles accompany a reported result. A scale profile plots
$T_n(\widehat e_{n,q})$ over $q$, a contribution profile ranks
$C_i(\widehat e_n^*)$, and a coverage profile records valid centres and
neighbourhood-size dispersion over radii. Together they distinguish a broad,
distributed departure from a result concentrated in one observation or one
narrow scale. They remain diagnostics. A large contribution can reveal a
state associated with rejection, not the economic mechanism that generated
it.

Primary implementation retains the raw maximum used in the matched
seven-method comparison. Studentisation replaces each fixed-radius statistic
with a resampling-centred and resampling-scaled version before maximisation.
Studentisation can equalise noisy and stable radii, but it changes the adaptation rule.
Ablation evidence shows modest average power gains and materially different
selected scales. Studentisation therefore remains a supplementary procedure
rather than a normalization silently substituted into the main test.

\subsection{Comparators}
\label{sec:benchmarks}

For a fixed radius, the local mean uses the uniform indicator kernel $K_e(z,z')=\1\{\|z-z'\|\le e\}$. MBCMI's contribution is the combined construction: observed-point centres, support weighting, coverage-screened quantile radii, joint resampling of the radius maximum, and a centre-by-scale diagnostic decomposition.

Fixed-count nearest-neighbour methods provide the closest geometric contrast. A $k$NN local-mean statistic fixes neighbourhood mass and lets geometric extent adapt to density. The NCMD graph test of \citet{chatterjee2026cmi} also uses nearest-neighbour structure and directly targets conditional-mean constancy. We report the author-specified $K=5$ and $K=10$ versions separately; no data-driven choice of $K$ is made. The canonical extension shows that NCMD is especially strong for interaction and checkerboard signals, while Ball gains for several broad local and radial alternatives.

MDC/MDD is an important global estimand-matched family \citep{shaozhang2014mdc,leeshao2018mddm}. The public author-code MDD implementation executed successfully but failed a prespecified finite-sample calibration screen: pooled pilot size was 6.625\%, one family rate was 7.625\%, and the worst cell was 10.5\%. Therefore MDD did not enter the canonical power ranking, and we did not retune it after observing that pilot. This exclusion is an implementation-comparability decision, not evidence that MDD is theoretically invalid or intrinsically inferior.

Distance correlation, HSIC, KSG, and MICe remain full-independence or dependence-strength benchmarks. Their targets differ from \eqref{eq:method-null}, so power comparisons are interpreted geometrically rather than as a single estimand-matched league table. Table~\ref{tab:literature-positioning} summarises the distinction.

\section{Inference and Theory}
\label{sec:inference}

Inference proceeds in two layers. First, population analysis establishes the
estimand and proves consistency for fixed alternatives visible at one or more
radii in the fixed finite grid. Second, resampling analysis studies null
calibration. For iid observations, a Ball-kernel reduction turns the
implemented raw radius vector into a finite collection of degenerate
V-statistics and gives direct Rademacher validity for their raw maximum under
primitive conditions. The serial bootstrap has a fixed algorithm and a
motivating finite-order process class. Stable autoregressive recolouring
preserves the relevant Ball-matrix leverage order, but the full
statistic-specific bootstrap proof after order selection and innovation
estimation remains open. We keep those distinctions visible.

\subsection{Asymptotic Framework}

Two population results define the target. Mean constancy implies
$Q(e)=0$ because iterated expectations give $r_e(z)=\mu p_e(z)$. Conversely,
regular local averaging recovers a nonconstant $m(z)$ on a positive-
probability subset as $e$ shrinks.

\begin{proposition}[Null implication and local identification]
\label{prop:main-identification}
Suppose $\E|Y|<\infty$. Null \eqref{eq:method-null} implies $Q(e)=0$ for
every radius with well-defined local means. Suppose additionally that the
support of $Z$ is compact and regular, its density is bounded above and away
from zero, and $m$ is bounded and continuous except on a null set. If
$\Pp\{m(Z)\ne\mu\}>0$, then some $e_0>0$ satisfies $Q(e)>0$ for every
$e\in(0,e_0)$.
\end{proposition}

Proposition~\ref{prop:main-identification} does not say that every radius detects every alternative.
A large radius can average positive and negative regions back to zero. A
very small radius can identify the population departure but lack enough
sample observations to estimate it. Multiscale search addresses this tension
by testing a finite range of supported radii rather than sending one
bandwidth to zero mechanically.

\paragraph{Fixed-grid limits.}

Assume iid observations with a finite $(4+\eta)$ outcome moment. Population-
standardised predictors have compact support satisfying a uniform interior
cone condition and a density bounded above and away from zero. Each indexed
population radius is positive, has uniformly positive local mass, and lies at
a continuity point with positive density of the pairwise-distance
distribution. A fixed minimum-neighbour count grows no faster than $o(n)$.

\begin{theorem}[Fixed-grid limit and grid-visible consistency]
\label{thm:main-fixed-grid}
Under these conditions, for every fixed $q\in\mathcal Q$,
\[
n^{-2}T_n(\widehat e_{n,q})\xrightarrow{p}Q(e_q),
\]
and convergence holds jointly over the 71 quantile indices. Coverage
screening is asymptotically inactive. Consequently,
\[
n^{-2}T_n^{\max}\xrightarrow{p}
\max_{q\in\mathcal Q}Q(e_q).
\]
If $Q(e_q)>0$ for at least one indexed radius and the
calibration-specific null critical value is $O_p(n)$, the corresponding test
is consistent against that fixed, grid-visible alternative.
\end{theorem}

Proof starts from a leave-one-out statistic. Euclidean balls form a VC
class, which gives uniform convergence of local probabilities and local
outcome numerators over centres and a fixed radius set. Positive local mass
controls the ratio that defines each local mean. Adding the centre back into
its own ball changes the normalised statistic by $o_p(1)$. Sample
standardisation and empirical distance quantiles add another $o_p(1)$ term.
Appendix~\ref{app:sample} gives the full argument.

Three restrictions matter. First, the grid contains 71 fixed indices; the
result does not cover a grid whose cardinality grows with $n$. Second,
regular support excludes singular laws concentrated on a lower-dimensional
manifold. Third, consistency is pointwise against fixed alternatives for which
$Q(e_q)>0$ at an indexed radius. Local identification at arbitrarily small
radii does not guarantee this grid-visibility condition because the smallest
population distance quantile stays positive as $n$ grows. The procedure is
therefore not an omnibus-consistent test against every nonconstant conditional
mean. Its adaptation operates over a deliberately finite family of supported
probability scales.

\paragraph{Local alternatives.}
The fixed-grid restriction also clarifies the power mechanism. Let
$Y_{n,i}=\mu+n^{-1/2}\delta_0(Z_i)+\varepsilon_i$, where
$\E[\delta_0(Z)]=0$ and $\E[\varepsilon_i\mid Z_i]=0$. For radius
$q$, define the Ball feature
\[
\Psi_q(c,z)=\sqrt{p_{e_q}(c)}
\left\{\frac{\1\{\|z-c\|\le e_q\}}{p_{e_q}(c)}-1\right\}
\]
and the local shift
\[
\Delta_q(c)=\E\{\Psi_q(c,Z)\delta_0(Z)\}.
\]
Its squared norm is exactly the population Ball criterion applied to the
local departure, $\|\Delta_q\|^2=Q_{\delta}(e_q)$. Positive and negative
pieces of $\delta_0$ can therefore cancel inside one radius while remaining
visible at another.

\begin{proposition}[Pitman local power on the fixed grid]
\label{prop:main-local-power}
Under the iid conditions used for Theorem~\ref{thm:main-rademacher}, suppose
$\delta_0$ is bounded and square-integrable. Jointly over the fixed radius
grid,
\[
\left\{n^{-1}T_n(\widehat e_{n,q}):q\in\mathcal Q\right\}
\Rightarrow
\left\{\|\mathbb G_q+\Delta_q\|^2:q\in\mathcal Q\right\},
\]
where $\{\mathbb G_q\}$ is the same centred Gaussian Hilbert-space limit
that generates the null second-order chaos. Conditional on the data, the
Rademacher vector under the local alternative converges to the centred null
vector, so its raw-max critical value converges to the null critical value.
If at least one $\Delta_q\ne0$ and the joint Gaussian law is nondegenerate
on the span of the shifts, the limiting rejection probability is strictly
larger than $\alpha$.
\end{proposition}

Proposition~\ref{prop:main-local-power} does not assert minimax optimality or
adaptation over shrinking spatial scales. It supplies the narrower connection
needed here: local power is governed by the projection of the mean departure
onto the finite collection of Ball features. Local-alternative calculations
play a similar diagnostic role in adaptive nonparametric specification tests
\citep{penaranda2022asset}. Appendix~\ref{app:sample} gives the proof and
shows why coherent local and radial signals can survive fixed-radius averaging
while rapidly sign-changing signals can cancel.

\paragraph{Overlapping neighbourhoods.}

Overlap does not alter the population criterion, but it changes the null
covariance. Writing $T_n(e)=u'M_{n,e}(X)u$ under the null shows that one
outcome enters several local means. Off-diagonal elements record shared
membership across balls; diagonal elements record self-inclusion and total
leverage. A calibration that treats centre contributions as independent
would ignore both features.

A raw fixed-radius statistic is typically $O_p(n)$ under the null and
$O_p(n^2)$ under a fixed alternative with $Q(e)>0$. This order separation
supports consistency but does not make the finite-sample null distribution
pivotal. Heteroskedasticity changes diagonal and off-diagonal quadratic-form
variation. Radius search adds dependence across the 71 forms. Bootstrap
must approximate their joint law, not only the marginal law at the observed
selected radius.

\subsection{IID Calibration}

Permutation $b$ reorders the outcome vector, holds predictor geometry fixed,
and recomputes the complete maximum. Its Monte Carlo $p$-value is
\begin{equation}
\widehat p_{\mathrm{perm}}
=
\frac{1+\sum_{b=1}^B
\1\{T_{n,b}^{\max}\ge T_n^{\max}\}}
{B+1}.
\label{eq:pvalue}
\end{equation}

\begin{proposition}[Finite-sample permutation validity]
\label{prop:perm}
If the conditional distribution of $(Y_1,\ldots,Y_n)$ given
$(X_1,\ldots,X_n)$ is invariant to every permutation under the null, the
randomisation test has conditional size no greater than $\alpha$.
\end{proposition}

Conditional exchangeability is stronger than mean independence. A process
$Y=\mu+\sigma(X)U$ can satisfy $\E[U\mid X]=0$ while changing the conditional
distribution of $Y$ across $X$. Unrestricted permutation then reallocates
large residual magnitudes across predictor states and can distort the null
law. Canonical matched simulations retain permutation because SNR zero
there generates full independence. A heteroskedastic extension evaluates
the broader conditional-mean null.

\paragraph{Rademacher multipliers.}

Let $\widehat u_i=Y_i-\bar Y$ and generate
\begin{equation}
Y_i^*=\bar Y+\xi_i\widehat u_i,
\qquad
\Pp(\xi_i=1)=\Pp(\xi_i=-1)=1/2.
\label{eq:wild}
\end{equation}
Independent signs preserve each residual magnitude and therefore preserve
the observed pattern of conditional scale. Predictor geometry remains fixed
because it depends only on $X$.

The primitive iid argument uses the population Ball kernel. For a centre
$c$ and observation $z$, define
\[
g_e(c,z)=\frac{\1\{\|z-c\|\le e\}}{p_e(c)}-1
\]
and
\begin{equation}
K_e(z,z')
=
\E\!\left[
p_e(C)g_e(C,z)g_e(C,z')
\right],
\label{eq:main-ball-kernel}
\end{equation}
where $C$ is an independent predictor draw. Uniform positive local mass
makes $K_e$ bounded. It is also positive semidefinite: it is an average of
rank-one centred Ball profiles.

Writing $u_i=Y_i-\mu$, the fixed-radius statistic satisfies
\[
n^{-1}T_n(\widehat e_{n,q})
=
\frac1n\sum_{r=1}^n\sum_{s=1}^n
K_{e_q}(Z_r,Z_s)u_ru_s+o_p(1)
\]
uniformly over the 71 fixed quantile indices. The sample Ball matrix
annihilates constants exactly, so replacing $\mu$ by $\bar Y$ does not
create an observed-statistic residual-estimation error. Sample
standardisation, empirical pairwise-distance quantiles, and the
minimum-neighbour/coverage rule enter the $o_p(1)$ remainder. Appendix
\ref{app:calibration} derives the matrix bounds and replacement argument.

Under the conditional-mean null, the kernel
\[
h_q\{(z,u),(z',u')\}
=
K_{e_q}(z,z')uu'
\]
is degenerate because $\E[u\mid Z]=0$. Hence the iid null problem is a
finite vector of degenerate V-statistics rather than 71 unrelated local
tests. A finite-rank spectral approximation plus the ordinary and
multiplier central limit theorems yields the following result; the argument
is closely related to weighted-bootstrap theory for degenerate U- and
V-statistics \citep{arconesgine1992bootstrap}.

\begin{theorem}[Primitive iid validity of the raw finite-grid maximum]
\label{thm:main-rademacher}
Suppose observations are iid; the predictor support, local-mass, and
pairwise-distance conditions used in Theorem~\ref{thm:main-fixed-grid} hold;
$\E[u\mid Z]=0$; conditional variances are bounded above and away from zero;
and conditional $(4+\eta)$ moments are uniformly bounded. Suppose at least
one indexed population Ball kernel is nonzero in $L^2(P_Z\times P_Z)$.
Then the conditional Rademacher law of
\[
\left\{
n^{-1}T_n^*(\widehat e_{n,q})
:q\in\mathcal Q
\right\}
\]
converges in probability, in bounded-Lipschitz distance, to the same
finite-dimensional second-order Gaussian-chaos law as
\[
\left\{
n^{-1}T_n(\widehat e_{n,q})
:q\in\mathcal Q
\right\}.
\]
Because at least one indexed kernel is nonzero, the limiting maximum has no positive atoms; the ideal Rademacher critical value for the raw maximum therefore has asymptotic level $\alpha$.
\end{theorem}

Theorem~\ref{thm:main-rademacher} is a result for the implemented raw
finite-grid maximum, not for a componentwise studentised surrogate. It does
not require a separate common-normalisation transfer assumption. The scope of Theorem ~\ref{thm:main-rademacher}
is nevertheless specific: the grid cardinality is fixed, predictor support
is regular and full-dimensional, and the result is iid. Theorem ~\ref{thm:main-rademacher} does not
cover a growing radius grid, singular predictor manifolds, or the serial
bootstrap used in the finance application. The 108-cell iid
heteroskedastic experiment remains useful as a finite-sample check rather
than as a substitute for the theorem.

\subsection{Time-Series Calibration}

Serial dependence creates a second preservation problem. iid signs retain
residual magnitudes but destroy autocorrelation. Circular shifts retain the
outcome's ordering but break its alignment with the predictor path, including
predictor-linked conditional variance. A dependent Gaussian multiplier can
approximate serial covariance, but one fixed dependence bandwidth need not
fit processes with different persistence and variance dynamics.

Serial calibration estimates a low-order outcome dynamic and perturbs its
dated innovations. This construction aims to preserve three features at once:
observed innovation magnitudes, innovation timing relative to predictors, and
the fitted serial propagation of shocks. This method is therefore more
structured than a model-free block or shift randomisation.

\paragraph{Bootstrap algorithm.}

Bootstrap implementation applies the following prespecified steps to each sample.

\begin{tcolorbox}[title={Prewhitened recursive Rademacher bootstrap},colback=white,colframe=black!55]
\begin{enumerate}[leftmargin=*,label=\arabic*.]
\item Centre the outcome and fit candidate autoregressions of orders
$p\in\{0,\ldots,6\}$ by least squares on the dated observations available to
that order. Exclude candidate fits that fail the implementation's stability
check.
\item Select $\widehat p$ by the comparable-sample BIC rule and retain
one-step innovations $\widehat\eta_t$ at their original dates for
$t=\widehat p+1,\ldots,n$.
\item Draw independent Rademacher signs $\xi_t^{(b)}$ and set
$\eta_t^{*(b)}=\xi_t^{(b)}\widehat\eta_t$ for those dated innovations.
\item Initialise the recursion with the observed centred values
$Y_1-\bar Y,\ldots,Y_{\widehat p}-\bar Y$ and generate
\[
Y_t^{*(b)}-\bar Y
=
\sum_{j=1}^{\widehat p}\widehat\phi_j
\{Y_{t-j}^{*(b)}-\bar Y\}
+
\eta_t^{*(b)},
\qquad t=\widehat p+1,\ldots,n.
\]
For $\widehat p=0$, set $Y_t^{*(b)}=\bar Y+\eta_t^{*(b)}$.
\item Recompute all fixed-radius statistics, coverage screening, and the raw
71-radius maximum.
\item Rank the observed maximum among $B$ bootstrap maxima using the
finite-resample correction in \eqref{eq:pvalue}.
\end{enumerate}
\end{tcolorbox}

Prewhitening removes fitted linear persistence from the outcome. Sign changes
preserve the absolute size and date of each fitted innovation. Recursive
recolouring restores the fitted propagation mechanism. Identical observed
initial values enter every bootstrap draw. BIC allows the order to vary across
samples without introducing an unrestricted sieve. Its implemented order set
is fixed, so the paper does not call the method a classical nonparametric
AR-sieve bootstrap. Implementation records document the numerical stability tolerance
used by the candidate-fit check; no theorem-level claim depends on a
particular tolerance.

\paragraph{Serial validity class.}

Write $u_t=Y_t-\mu$ and suppose
\[
u_t=\sum_{j=1}^{p_0}\phi_j u_{t-j}+\eta_t,
\qquad p_0\le6,
\]
with a stable autoregressive polynomial. Formal validity of the recursive
sign bootstrap uses a conditional sign-symmetry class tailored to the
algorithm. For $t>p_0$, write $\eta_t=a_t s_t$, $a_t\ge0$, and let
$\mathcal A_n$ contain the complete predictor path, the observed initial
conditions, and the innovation magnitudes. Conditional on $\mathcal A_n$,
the signs $s_t$ are independent Rademacher variables. Magnitudes may vary
with predictors and volatility states, so conditional heteroskedasticity is
unrestricted apart from moment and nondegeneracy bounds. Sign symmetry is
stronger than the contemporaneous null in \eqref{eq:method-null}; it is a
maintained calibration condition, not a consequence of mean independence.

Known autoregressive parameters make the logic transparent. Conditional on
$\mathcal A_n$, the observed trajectory is generated by one Rademacher sign
vector passed through the stable AR filter. Redrawing those signs and
recursively recolouring the innovations therefore reproduces the exact
conditional oracle law of the complete 71-radius vector and its maximum. The
feasible algorithm estimates the order, coefficients, and innovation
magnitudes. Over the fixed order set, standard finite-order BIC separation
selects $p_0$ consistently; least-squares coefficient and residual replacement
then alter the $n^{-1}$ Ball vector by $o_p(1)$. Appendix~\ref{app:serial}
proves each replacement step.

\begin{theorem}[Feasible recursive-sign validity for stable finite-order dynamics]
\label{thm:main-serial}
Suppose the stable AR model above holds; innovations are conditionally
sign-symmetric given $\mathcal A_n$; conditional innovation moments of order
$4+\delta$ are uniformly bounded and second moments are bounded away from
zero; the true order is uniquely represented in the fixed candidate set;
and the fixed-grid Ball matrices satisfy the leverage conditions established
in Appendix~\ref{app:calibration}. Then the comparable-sample BIC selects
$p_0$ with probability approaching one, the least-squares AR estimator is
root-$n$ consistent, and the conditional law of the feasible prewhitened
recursive Rademacher radius vector converges in bounded-Lipschitz distance to
the oracle conditional sign-randomisation law. If the largest conditional atom of the oracle raw maximum converges to zero in probability, the feasible raw-max test has asymptotic conditional level $\alpha$ and hence unconditional level $\alpha$.
\end{theorem}

If the anti-atom condition fails, the oracle plus-one rank test remains conditionally conservative and the feasible approximation inherits that conservative interpretation. The anti-atom condition is needed for equality with the nominal asymptotic rejection probability, not for the basic randomisation ordering.

Theorem~\ref{thm:main-serial} proves validity for a restricted serial class
that matches the recursive sign mechanism. It allows predictor-linked
innovation magnitudes but excludes asymmetric innovation signs, unstable or
long-memory dynamics, and feedback from current innovation signs into the
conditioned predictor path. For the broader martingale-difference working
class, let $\mathcal G_t=\sigma(\mathcal F_{t-1},X_t)$ and impose
$\E[\eta_t\mid\mathcal G_t]=0$ together with
$\sum_{j=1}^{p_0}\phi_j\E[u_{t-j}\mid X_t]=0$. Those restrictions preserve
the contemporaneous conditional-mean target but do not imply conditional
sign symmetry. Application-aligned null simulations therefore complement,
rather than enlarge, the theorem: they evaluate the implemented procedure in
asymmetric and persistence/variance configurations for which no general
statistic-specific result is claimed.

Circular shifts and block permutations remain sensitivity checks. The paper
does not claim model-free validity for arbitrary stationary processes,
unrestricted conditional-independence testing, or exact finite-sample
validity after AR estimation. The serial theorem instead supplies a clean
formal core for the algorithm and leaves the broader simulation evidence in
its proper finite-sample role.

\subsection{Scope and Limitations}

Five boundaries govern use of the MBCMI implementation procedure. Predictor dimension can make every local method sparse, and distance-quantile indexing does not remove the curse of dimensionality. Raw local averaging can lose power when signs alternate inside supported balls. Fixed-grid adaptation is not omnibus consistency: the procedure is consistent only when at least one indexed probability scale carries a nonzero Ball-smoothed mean departure. Selected radii describe the maximising scale inside that finite domain and do not estimate a unique structural signal width. Predictor geometry also matters because the robustness experiment shows no universally optimal transformation.

Primitive iid theory uses compact regular support and uniform local mass. Those assumptions deliberately simplify observed-centre and empirical-radius control and become demanding as dimension grows; the $d=10$ and $d=20$ experiments therefore provide finite-sample evidence rather than a claim that the local-mass condition is innocuous in moderate dimension. Gaussian, $t$, and lognormal predictor laws in the robustness experiment also lie outside the compact-support theorem and provide evidence beyond its formal scope. The serial theorem likewise has a sharp boundary. It covers stable finite-order autoregressions with conditionally sign-symmetric innovations and predictor-linked magnitudes, while the broader martingale-difference class remains simulation-validated.

Endpoint diagnostics showed that a $q\le .50$ search was empirically binding in a material subset of application-sized windows. The $q\le .75$ extension passes the application-aligned serial-null criterion. A sparse $q\le .90$ tail sensitivity does not move any central selected radius above .75, but it changes the Holm-adjusted macro-family status of unemployment-rate change; the paper therefore reports that sensitivity explicitly. High-contribution regions identify predictor states associated with the statistic, not causal regimes.

Appendix~\ref{app:proof-audit} makes the evidential hierarchy explicit. Appendices~\ref{app:population}--\ref{app:calibration} contain the population, fixed-grid, local-power, and iid calibration proofs. Appendix~\ref{app:serial} proves the restricted feasible serial result and records the broader simulation scope. Later appendices document implementation, canonical comparisons, predictor-law robustness, finance residualisation, weighting, and reproducibility.

\section{Monte Carlo Evidence}
\label{sec:mc}

\subsection{Canonical design}

The canonical experiment uses nine conditional-mean DGPs, $n\in\{200,500\}$, $d\in\{2,10,20\}$, two active coordinates, 21 SNR values from zero to two, and 1,000 common datasets per cell. The 54 DGP--sample-size--dimension blocks generate 1,134 design--SNR cells. Table~\ref{tab:simulation_scope} summarises the design and Table~\ref{tab:dgp_taxonomy} names the geometries.

\begin{table}[!htbp]
\centering
\caption{Scope of the canonical matched Monte Carlo experiment}
\label{tab:simulation_scope}
\begin{threeparttable}
\small
\begin{tabularx}{\textwidth}{@{}lX@{}}
\toprule
Design component & Specification \\
\midrule
Alternative geometries & Nine DGPs spanning global smooth, sign-changing,
compact local, radial, disconnected, alternating, narrow-support, and
local-magnitude conditional-mean departures. \\
Sample size & $n\in\{200,500\}$. \\
Predictor dimension & $d\in\{2,10,20\}$ with two active coordinates. \\
Signal strength & 21 SNR values from 0 to 2. \\
Reported columns & MBCMI, distance correlation, two HSIC implementations, KSG, $k$NN local means, and MICe in the historical seven-method archive; NCMD $K=5$ and $K=10$ enter a separate paired canonical extension. \\
Monte Carlo replication & 1,000 common datasets per method--design--SNR cell. \\
Matched comparison structure & 54 design blocks, 1,134 design--SNR cells,
and paired method comparisons on identical simulated datasets. \\
Primary reporting principle & Size under the null, geometry-specific power,
scale-adaptation regret, and explicit failure cases; no average-power
dominance claim. \\
\bottomrule
\end{tabularx}
\begin{tablenotes}[flushleft]\footnotesize
\item Notes: A design block is a DGP--sample-size--dimension combination.
The replication archive contains the complete cell-level results.
\end{tablenotes}
\end{threeparttable}
\end{table}

\begin{table}[!htbp]\centering
\caption{Data-generating processes in the matched Monte Carlo design}
\label{tab:dgp_taxonomy}
\begin{tabularx}{\textwidth}{@{}p{0.17\textwidth}p{0.19\textwidth}X@{}}
\toprule
DGP & Family & Structural feature \\
\midrule
Linear & Global & Smooth linear index in the active coordinates. \\
Interaction & Global & Sign-changing product interaction in the active coordinates. \\
Local island & Local/geometric & A compact region carries a mean shift while the remainder of the support does not. \\
Local linear island & Local/geometric & A smooth linear signal operates only inside a compact region. \\
Ring signal & Local/geometric & An annular, non-convex region carries the signal. \\
Two islands & Local/geometric & Two separated compact regions carry opposing local signals. \\
Checkerboard & Local/geometric & Neighbouring regions alternate in sign. \\
Thin band & Local/geometric & Signal is concentrated in a narrow diagonal band. \\
Local absolute signal & Local/geometric & An absolute-value conditional-mean signal operates only inside a compact region. \\
\bottomrule
\end{tabularx}
\end{table}

Figure~\ref{fig:dgp-geometry} shows the signal shapes. These surfaces define conditional-mean alternatives, not predictor distributions.

\begin{figure}[!htbp]
\centering
\begin{subfigure}{0.31\textwidth}\includegraphics[width=\linewidth]{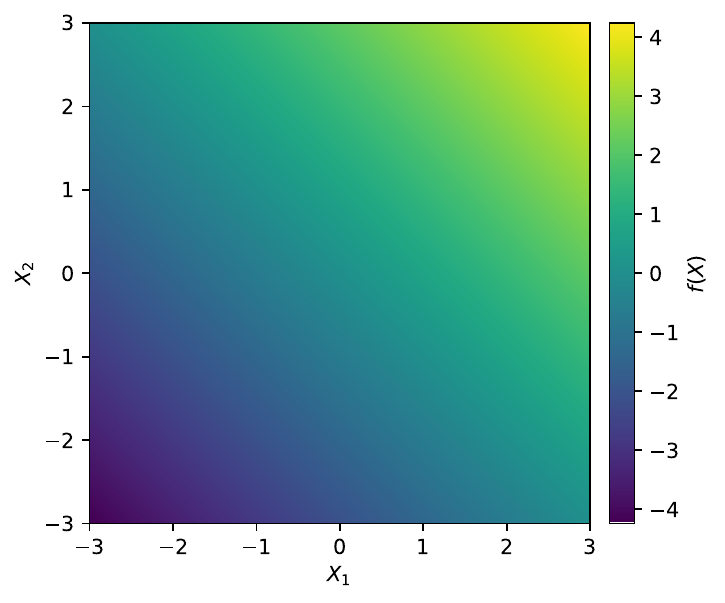}\caption{Linear}\end{subfigure}
\begin{subfigure}{0.31\textwidth}\includegraphics[width=\linewidth]{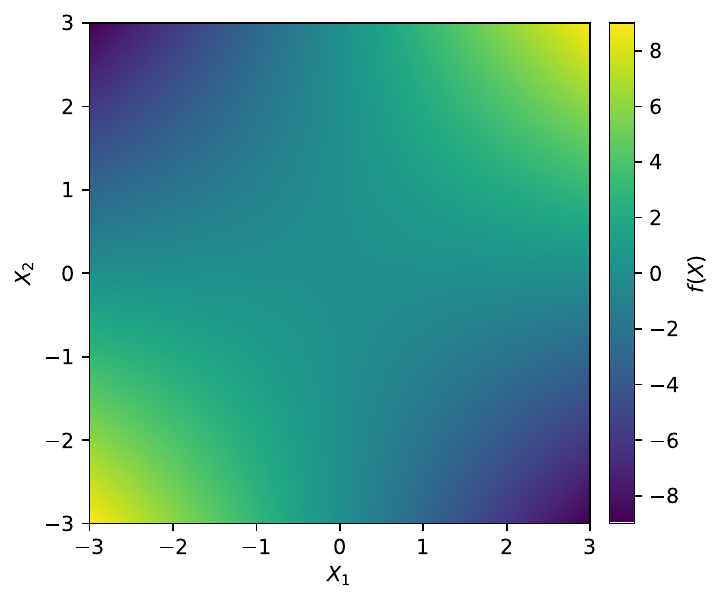}\caption{Interaction}\end{subfigure}
\begin{subfigure}{0.31\textwidth}\includegraphics[width=\linewidth]{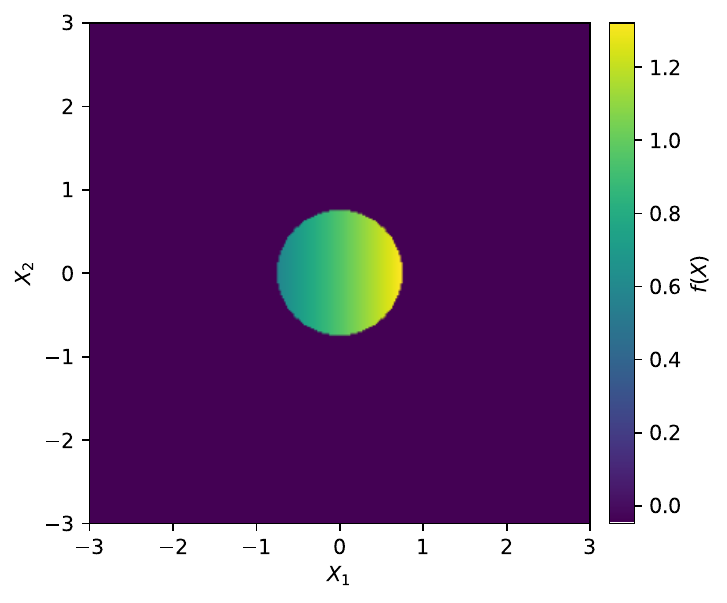}\caption{Local island}\end{subfigure}
\begin{subfigure}{0.31\textwidth}\includegraphics[width=\linewidth]{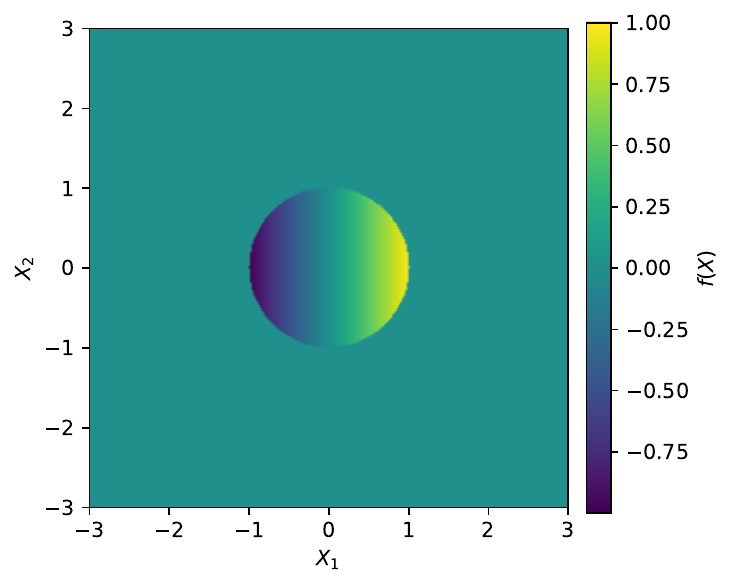}\caption{Local linear island}\end{subfigure}
\begin{subfigure}{0.31\textwidth}\includegraphics[width=\linewidth]{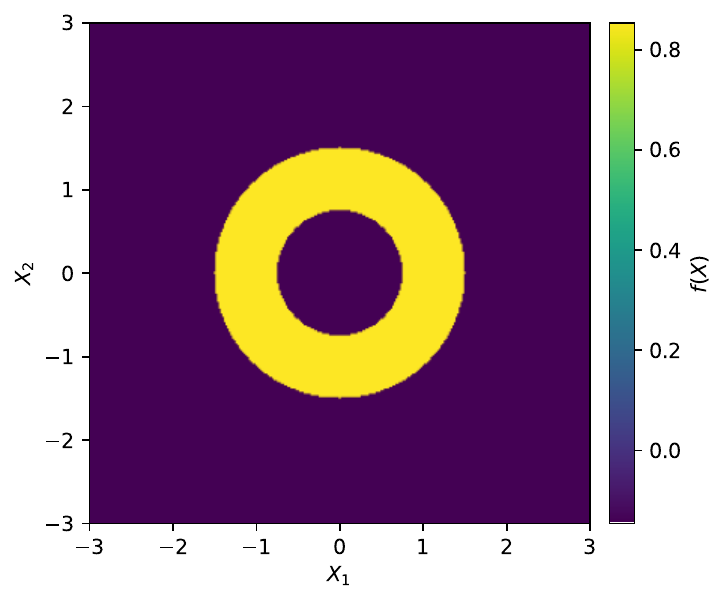}\caption{Ring}\end{subfigure}
\begin{subfigure}{0.31\textwidth}\includegraphics[width=\linewidth]{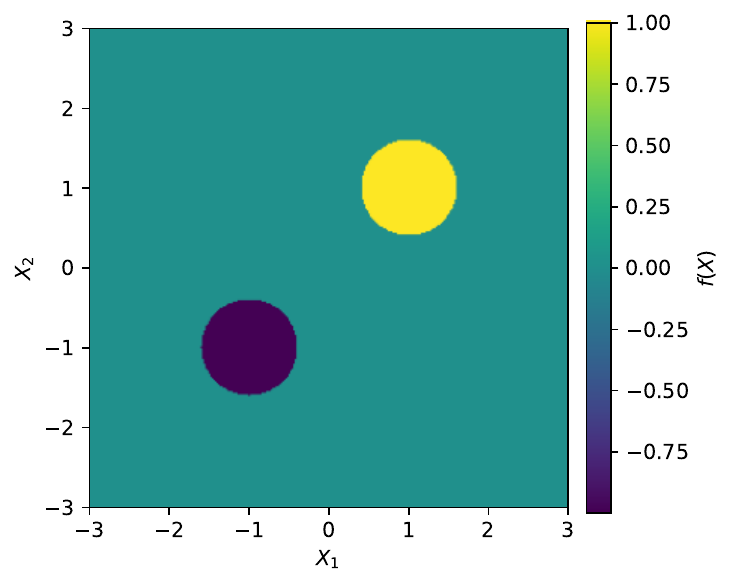}\caption{Two islands}\end{subfigure}
\begin{subfigure}{0.31\textwidth}\includegraphics[width=\linewidth]{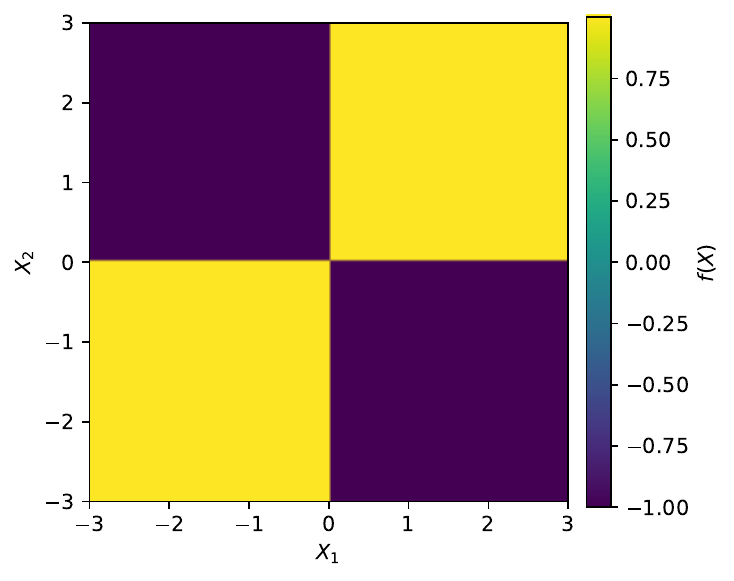}\caption{Checkerboard}\end{subfigure}
\begin{subfigure}{0.31\textwidth}\includegraphics[width=\linewidth]{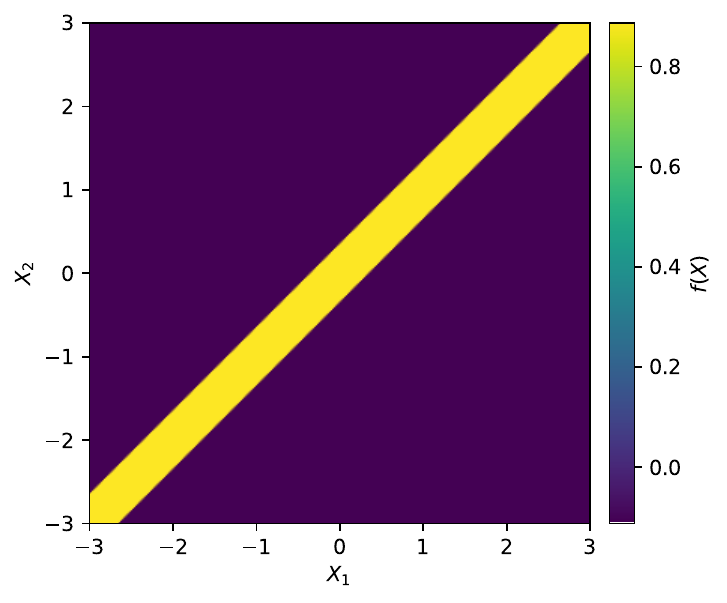}\caption{Thin band}\end{subfigure}
\begin{subfigure}{0.31\textwidth}\includegraphics[width=\linewidth]{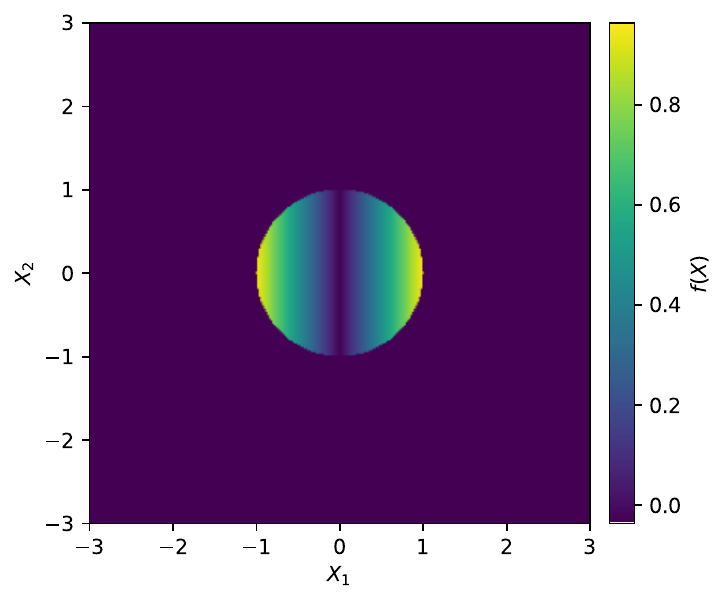}\caption{Local absolute signal}\end{subfigure}
\caption{Conditional-mean geometries in the canonical Monte Carlo design.}
\label{fig:dgp-geometry}
\end{figure}

\subsection{Calibration and the geometry-specific power envelope}

Under the canonical independence null, pooled rejection is 5.18\% for MBCMI $q\le .75$, 5.17\% for NCMD $K=5$, and 5.13\% for NCMD $K=10$. The corresponding 54-cell ranges are 3.8--6.7\%, 3.7--7.8\%, and 3.1--7.1\%. These results support comparison of power on the paired canonical datasets rather than adjustment for obvious size failure.

Low-signal behaviour carries more information than rankings after power has saturated. Table~\ref{tab:low_snr_auc_main_design} therefore reports normalised AUC only over $0\le\mathrm{SNR}\le0.5$ for the main $n=500,d=10$ design. Figure~\ref{fig:representative-power} shows three representative power curves: a conventional linear case, a local-island case favourable to the fixed-radius averaging in MBCMI, and a checkerboard case in which rapid sign changes are adverse to that geometry. The separate NCMD extension then supplies the closest estimand-matched comparison on the same canonical datasets.

\begin{table}[!htbp]\centering
\caption{Normalised low-signal power AUC over $0\leq\mathrm{SNR}\leq 0.5$ ($n=500$, $d=10$, $s=2$)}
\label{tab:low_snr_auc_main_design}
\begin{tabular}{lrrrrrrr}
\toprule
DGP & MBCMI & dCor & HSIC & HSIC-perm & KSG & kNN & MICe \\
\midrule
Checkerboard & 0.156 & 0.099 & 0.130 & 0.092 & 0.209 & \textbf{0.444} & 0.056 \\
Interaction & 0.182 & 0.206 & 0.262 & 0.186 & \textbf{0.468} & 0.303 & 0.122 \\
Linear & 0.913 & \textbf{0.952} & 0.922 & 0.941 & 0.742 & 0.876 & 0.079 \\
Local island & 0.617 & 0.116 & 0.152 & 0.105 & 0.083 & \textbf{0.629} & 0.197 \\
Local linear island & 0.557 & 0.445 & 0.281 & 0.322 & 0.056 & \textbf{0.626} & 0.065 \\
Local absolute signal & \textbf{0.512} & 0.100 & 0.127 & 0.095 & 0.068 & 0.438 & 0.146 \\
Ring signal & \textbf{0.179} & 0.061 & 0.069 & 0.065 & 0.031 & 0.095 & 0.085 \\
Thin band & 0.215 & 0.068 & 0.075 & 0.067 & 0.083 & \textbf{0.267} & 0.079 \\
Two islands & 0.717 & \textbf{0.744} & 0.494 & 0.607 & 0.214 & 0.650 & 0.061 \\
\bottomrule
\end{tabular}
\begin{minipage}{0.96\textwidth}\footnotesize
Notes: Normalised AUC integrates rejection frequency over the displayed SNR range.
Bold entries identify the numerically largest value within each DGP.
The two HSIC columns are legacy implementation sensitivities of one kernel
architecture. MICe receives only the radial scalar $\|X-\bar X\|$ and is
not used as a headline multivariate comparator. Every power estimate uses
1,000 matched replications.
\end{minipage}
\end{table}

\begin{figure}[!htbp]
\centering
\begin{subfigure}{0.32\textwidth}
\includegraphics[width=\linewidth]{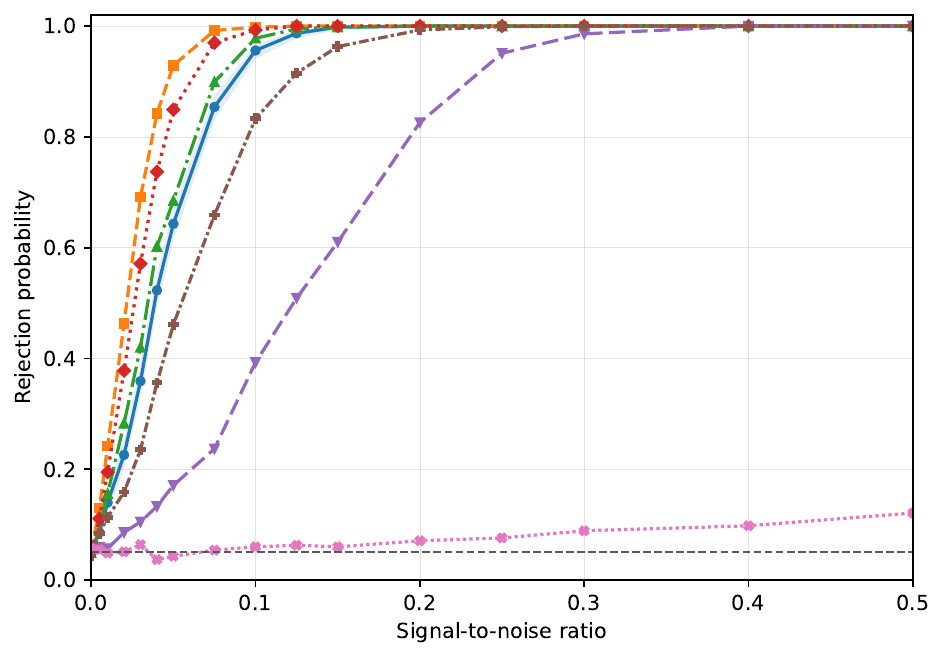}
\caption{Linear}
\end{subfigure}
\begin{subfigure}{0.32\textwidth}
\includegraphics[width=\linewidth]{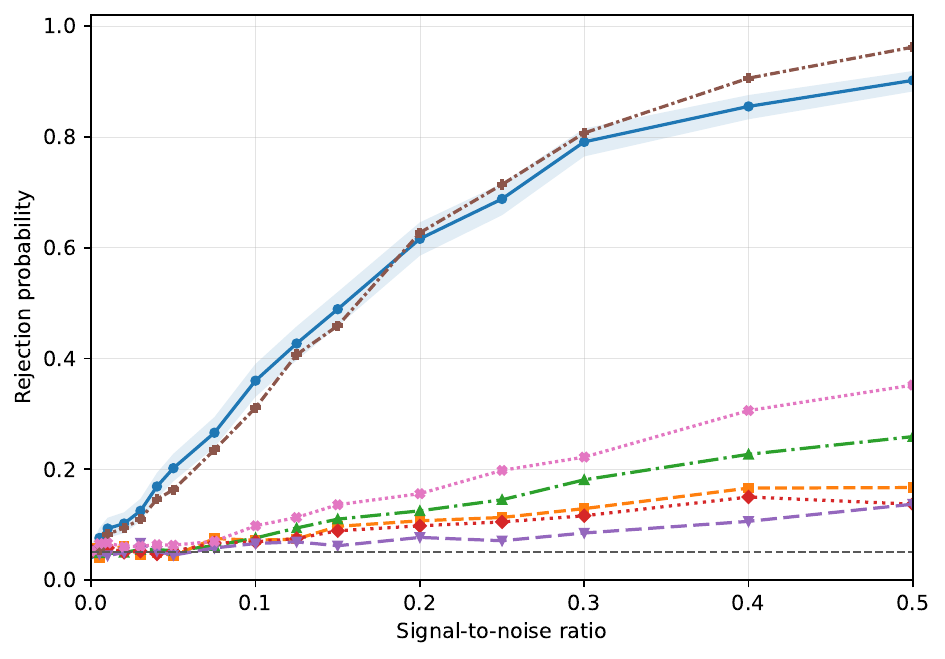}
\caption{Local island}
\end{subfigure}
\begin{subfigure}{0.32\textwidth}
\includegraphics[width=\linewidth]{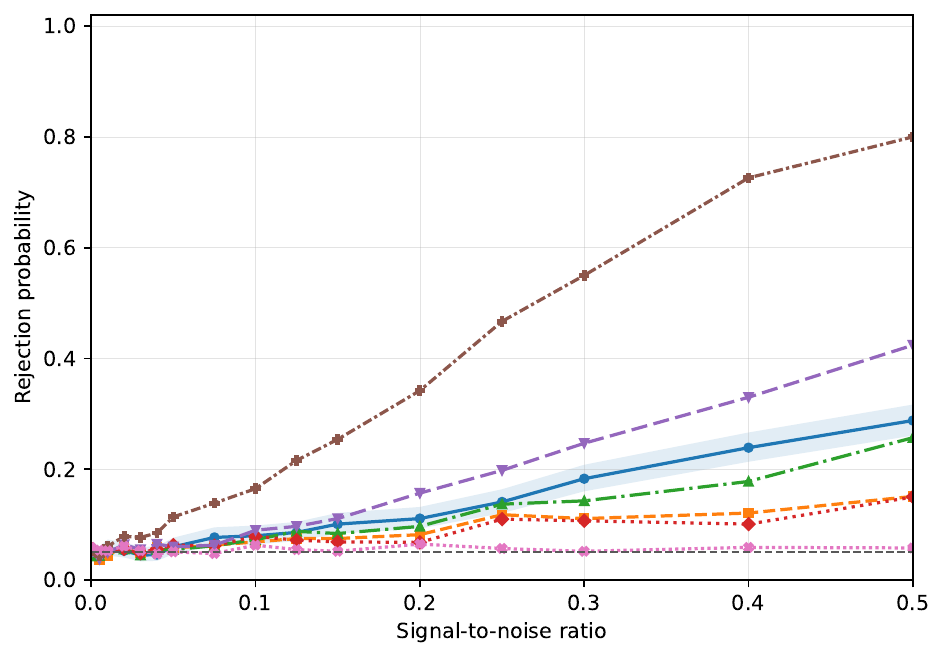}
\caption{Checkerboard}
\end{subfigure}
\vspace{0.2em}
\includegraphics[width=0.78\textwidth]{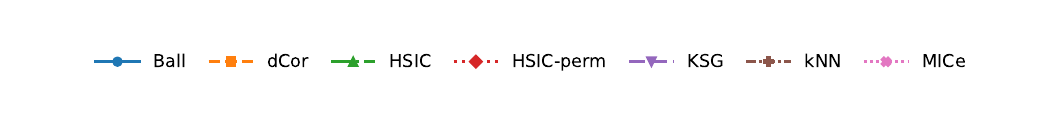}
\caption{Representative low-signal power curves in the canonical $n=500,d=10,s=2$ design. Panels compare MBCMI (Ball) with seven established benchmarks on matched simulations. Table~\ref{tab:ncmd-dgp-auc} and Figure~\ref{fig:ncmd-paired-difference} report the closest NCMD comparison on exact paired datasets.}
\label{fig:representative-power}
\end{figure}

\begin{table}[!htbp]
\centering
\caption{Geometry-specific canonical power: MBCMI and NCMD}
\label{tab:ncmd-dgp-auc}
\begin{threeparttable}
\small
\begin{tabular}{lrrr}
\toprule
DGP & MBCMI $q\le .75$ & NCMD $K=5$ & NCMD $K=10$ \\
\midrule
Linear & 0.955 & 0.941 & 0.954 \\
Interaction & 0.441 & 0.819 & 0.842 \\
Local island & 0.745 & 0.623 & 0.640 \\
Local linear island & 0.705 & 0.628 & 0.666 \\
Ring signal & 0.457 & 0.366 & 0.364 \\
Two islands & 0.803 & 0.754 & 0.804 \\
Checkerboard & 0.451 & 0.693 & 0.714 \\
Thin band & 0.492 & 0.458 & 0.474 \\
Local absolute signal & 0.684 & 0.534 & 0.562 \\
\bottomrule
\end{tabular}
\begin{tablenotes}[flushleft]\footnotesize
\item Notes: Entries are mean normalised full-range AUC across the common sample-size and dimension blocks. NCMD $K=5$ and $K=10$ are separate author-specified choices; no data-driven selection between them is made. All comparisons use the exact canonical datasets and seed map. The table is intended to show a geometry-specific power envelope rather than an overall ranking.
\end{tablenotes}
\end{threeparttable}
\end{table}

AUC is the trapezoidal integral of the rejection-probability curve over the reported SNR range, divided by that range, so it lies on the same zero-to-one scale as power. Table~\ref{tab:ncmd-dgp-auc} rejects a universal-ranking narrative. NCMD dominates the interaction and checkerboard designs. MBCMI leads for local island, local linear island, ring, thin band, and local absolute signal, and the two approaches are essentially tied in full-range linear and two-island AUC. At $n=500,d=10,$ SNR $=.25$, NCMD $K=10$ exceeds MBCMI by 72.8 percentage points in interaction and 53.8 points in checkerboard, whereas MBCMI exceeds NCMD $K=10$ by 16.4 points in local island, 10.6 in ring, and 23.7 in local absolute signal. MBCMI's advantage is therefore tied to spatial geometry rather than average rank.

Figure~\ref{fig:ncmd-paired-difference} reports paired rejection-rate differences for three representative geometries in the main $n=500,d=10,s=2$ design. Negative differences favour MBCMI and positive differences favour NCMD. The figure makes the trade-off visible without compressing performance into one average: MBCMI has a large low-SNR advantage in the linear case, retains an advantage for the local island, and loses sharply to NCMD as checkerboard signal strengthens.

\begin{figure}[!htbp]
\centering
\begin{subfigure}{0.32\textwidth}
\includegraphics[width=\linewidth]{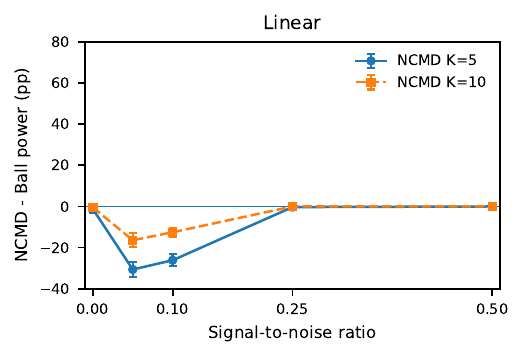}
\caption{Linear}
\end{subfigure}
\begin{subfigure}{0.32\textwidth}
\includegraphics[width=\linewidth]{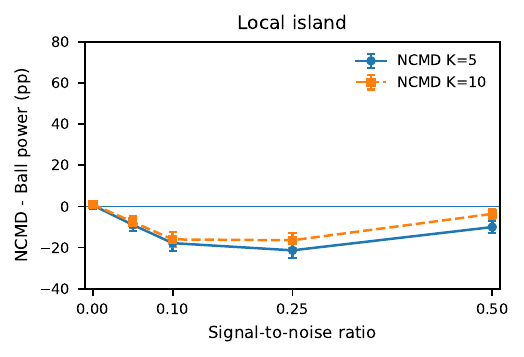}
\caption{Local island}
\end{subfigure}
\begin{subfigure}{0.32\textwidth}
\includegraphics[width=\linewidth]{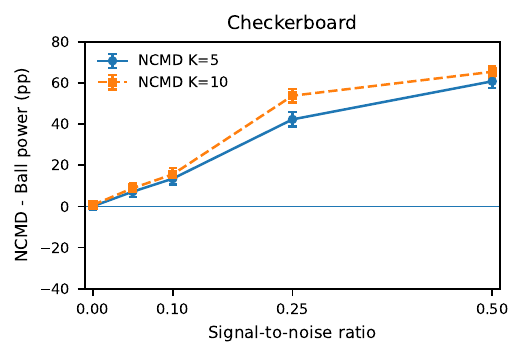}
\caption{Checkerboard}
\end{subfigure}
\caption{Paired power differences between NCMD and MBCMI(Ball). Points show NCMD minus MBCMI rejection rates in percentage points; bars are paired 95\% Monte Carlo intervals. Negative values favour MBCMI. Both author-specified NCMD choices are reported without data-driven selection of $K$.}
\label{fig:ncmd-paired-difference}
\end{figure}

The extended nine-method averages reinforce that point. NCMD $K=10$ has the largest mean full-range AUC (0.669), followed by $k$NN (0.654), NCMD $K=5$ (0.646), and MBCMI (0.637). MBCMI's low-SNR average (0.465) exceeds NCMD $K=5$ and remains close to $k$NN and NCMD $K=10$. Appendix~\ref{app:canonical-extension} reports the full summary. The NCMD results form a paired extension on the same canonical datasets and seed map used for the seven established comparators.

\subsection{Radius aggregation and serial calibration}

Radius-domain calibration matters because the shorter $q\le .50$ grid bound in a substantial subset of rolling samples. Extending the dense grid to .75 passes the application-aligned null study. Studentising radius-specific statistics does not dominate the raw maximum.

\begin{table}[!htbp]
\centering
\caption{Application-aligned serial-null calibration for the final $q\le .75$ search}
\label{tab:serial-calibration}
\begin{threeparttable}
\small
\begin{tabular}{lccc}
\toprule
Aggregation & Pooled rejection & Exact 95\% interval & Cell range \\
\midrule
Raw maximum (primary) & 0.0425 & [0.0394, 0.0457] & 0.026--0.058 \\
Studentised maximum & 0.0394 & [0.0364, 0.0425] & 0.022--0.056 \\
\bottomrule
\end{tabular}
\begin{tablenotes}[flushleft]\footnotesize
\item Notes: Sixteen application-aligned serial null cells use $n=549$, $d\in\{5,6\}$, $R=1000$, and $B=999$. The raw maximum passes the prespecified calibration gate. Studentisation is mildly conservative and triggers the prespecified systematic-under-rejection diagnostic. Both variants repeat the complete 71-radius search inside each bootstrap draw.
\end{tablenotes}
\end{threeparttable}
\end{table}

Studentisation slightly raises mean positive-SNR power in the canonical ablation from 31.55\% to 32.06\%, with 23 cell wins, 22 losses, and three ties. The application-aligned serial null is more consequential: pooled size falls below the prespecified 4\% lower bound, three cells are below 3\%, and selected scales shift upward. Mean selected $q$ rises from .192 for the raw statistic to .325 after studentisation, and the studentised statistic selects the .75 boundary in 13.6\% of serial-null replications. We therefore retain the raw maximum and interpret studentisation as an adverse normalisation check rather than a default improvement.

\subsection{Predictor-law and metric robustness}

The canonical predictor distribution is not the only geometry considered. A targeted experiment fixes $n=500,d=10,s=2$ and evaluates six predictor laws: independent Gaussian, AR(1)-correlated Gaussian with $\rho=.7$, correlated standardised $t_5$, centred standardised lognormal, a two-component Gaussian mixture, and 2\% scale contamination. Each law is paired with linear, local-island, ring, local-absolute, thin-band, and checkerboard signals over SNR $0,.05,.10,.25,.50$.

\begin{table}[!htbp]
\centering
\caption{Predictor-law and geometry robustness of the raw $q\le .75$ MBCMI test}
\label{tab:predictor-robustness}
\begin{threeparttable}
\small
\begin{tabular}{lrrrr}
\toprule
Predictor law & z-score null size & z-score power & MAD power & Whitening power \\
\midrule
Independent Gaussian & 0.0528 & 0.3915 & 0.3923 & 0.3911 \\
Correlated Gaussian & 0.0490 & 0.4660 & 0.4656 & 0.3801 \\
Correlated $t_5$ & 0.0517 & 0.5256 & 0.5265 & 0.4496 \\
Lognormal & 0.0507 & 0.4781 & 0.4803 & 0.4770 \\
Gaussian mixture & 0.0533 & 0.3553 & 0.3100 & 0.3643 \\
Contaminated Gaussian & 0.0537 & 0.3929 & 0.3951 & 0.3813 \\
\bottomrule
\end{tabular}
\begin{tablenotes}[flushleft]\footnotesize
\item Notes: $n=500$, $d=10$, two active coordinates, $R=1000$. Power averages the positive SNR values $0.05,0.10,0.25,0.50$ over six DGPs. z-score + Euclidean is primary. MAD uses median/MAD scaling; whitening uses the sample covariance. There were zero mandatory geometry failures. Power and the best geometry remain predictor-law specific.
\end{tablenotes}
\end{threeparttable}
\end{table}

Raw MBCMI remains approximately calibrated across the predictor laws, and none of the mandatory geometries fails computationally. z-score and robust-MAD scaling are similar for several laws. Whitening is not a universal correction: it lowers average positive-SNR power by about 8.6 points under correlated Gaussian predictors and 7.7 points under correlated $t_5$, yet modestly improves the mixture relative to z-scoring. The correct robustness claim is therefore limited: MBCMI's principal local/radial advantages are not an artefact of independent Gaussian predictors, but metric choice remains distribution-specific.

\subsection{Robustness Summary}

The experiments support three claims. First, the raw $q\le .75$ maximum is well calibrated in the studied iid and application-aligned serial nulls. Second, multiscale fixed-radius adaptation is valuable for broad local and radial alternatives, but sign-changing and fixed-count-friendly geometries can favour NCMD or $k$NN. Third, predictor-distribution changes do not erase the method's basic strengths, although no transformation dominates. They do not prove model-free serial validity or universal power dominance.

\section{Finance Applications}
\label{sec:finance}

\subsection{Data and design}

The application uses monthly observations from January 1980 through September 2025. The factor system contains market excess return, SMB, HML, RMW, CMA, and Momentum. Each factor becomes the outcome once, with the other five factors as contemporaneous predictors. Four macro-finance equations use market excess return, industrial-production growth, unemployment-rate change, and Baa-spread change as outcomes with the contemporaneous state vectors in Table~\ref{tab:finance-design}. Complete-case alignment yields $n=549$ for every full-sample equation.

\begin{table}[!htbp]
\centering
\caption{Finance outcomes, transformations, and contemporaneous predictor sets}
\label{tab:finance-design}
\begin{threeparttable}
\small
\begin{tabularx}{\textwidth}{@{}p{0.22\textwidth}p{0.27\textwidth}X@{}}
\toprule
Outcome & Transformation and units & Contemporaneous predictors \\
\midrule
\multicolumn{3}{l}{\textit{Panel A: U.S. equity factors}}\\
Market, SMB, HML, RMW, CMA, or Momentum & Monthly percentage-point factor returns & The other five factor returns. \\
\addlinespace
\multicolumn{3}{l}{\textit{Panel B: Macro-finance equations}}\\
Market excess return & Monthly percentage points & \texttt{FEDFUNDS}, inflation, \texttt{UNRATE}, \texttt{T10Y2YM}, \texttt{BAA10YM}, industrial-production growth. \\
Industrial-production growth & $1200\,\Delta\log(\texttt{INDPRO})$, annualised percent & Market excess return, \texttt{FEDFUNDS}, inflation, \texttt{UNRATE}, \texttt{T10Y2YM}, \texttt{BAA10YM}. \\
Unemployment-rate change & $\Delta\texttt{UNRATE}$, percentage points & Market excess return, \texttt{FEDFUNDS}, inflation, \texttt{T10Y2YM}, \texttt{BAA10YM}, industrial-production growth. \\
Baa-spread change & $\Delta\texttt{BAA10YM}$, percentage points & Market excess return, \texttt{FEDFUNDS}, inflation, \texttt{UNRATE}, \texttt{T10Y2YM}, industrial-production growth. \\
\bottomrule
\end{tabularx}
\begin{tablenotes}[flushleft]\footnotesize
\item Notes: Factor returns come from the Kenneth R. French Data Library. Macro series come from FRED. Inflation equals $1200\,\Delta\log(\texttt{CPIAUCSL})$. Complete-case alignment follows all transformations, no values are imputed, and the common sample runs from January 1980 through September 2025. Predictors are standardised within each full sample or rolling window. These equations are contemporaneous and do not test forecasting. Appendix~\ref{app:finance-robustness} gives exact source identifiers and data-construction details.
\end{tablenotes}
\end{threeparttable}
\end{table}

Every primary MBCMI calculation uses z-scored predictors, the raw $q=.05,.06,\ldots,.75$ maximum, support weights $N_i$, 20\% radius coverage, and $B=999$ prewhitened recursive Rademacher draws. Holm adjustment is applied separately to the six factor equations and four macro equations. Rolling analysis uses 240-month windows with a one-month step and recomputes the complete procedure inside each window.

\subsection{Full-sample dependence and the linear benchmark}

\begin{table}[!htbp]
\centering
\caption{Full-sample Ball results and cross-fitted linear benchmark}
\label{tab:full-sample-finance}
\begin{threeparttable}
\small
\begin{tabularx}{\textwidth}{@{}lYrrrrr@{}}
\toprule
Family & Outcome & MBCMI $p$ & Holm $p$ & Selected $q$ & CF $R^2$ & Residual MBCMI $p$ \\
\midrule
Factors & MKT & 0.001 & 0.006 & 0.42 & 0.102 & 0.182 \\
 & SMB & 0.011 & 0.033 & 0.24 & 0.060 & 0.332 \\
 & HML & 0.001 & 0.006 & 0.42 & 0.404 & 0.914 \\
 & RMW & 0.062 & 0.062 & 0.27 & -0.056 & 0.072 \\
 & CMA & 0.001 & 0.006 & 0.48 & 0.493 & 0.481 \\
 & Momentum & 0.025 & 0.050 & 0.48 & 0.035 & 0.080 \\
\midrule
Macro & Market excess return & 0.825 & 0.825 & 0.10 & -0.009 & 0.123 \\
 & Industrial-production growth & 0.094 & 0.282 & 0.68 & 0.009 & 0.923 \\
 & Unemployment-rate change & 0.007 & 0.028 & 0.75 & 0.088 & 1.000 \\
 & Baa-spread change & 0.203 & 0.406 & 0.42 & 0.027 & 0.198 \\
\bottomrule
\end{tabularx}
\begin{tablenotes}[flushleft]\footnotesize
\item Notes: MBCMI uses the raw maximum over 71 quantile radii and the prewhitened recursive Rademacher bootstrap with $B=999$. Holm adjustment is within the six factor equations and four macro equations. CF $R^2$ is five-fold contiguous time-ordered cross-fitted OLS $R^2$. Residual MBCMI applies the complete bootstrap procedure to cross-fitted residuals and refits the nuisance OLS model inside every bootstrap draw. None of the ten full-sample residual tests rejects at 5\%.
\end{tablenotes}
\end{threeparttable}
\end{table}

Five of six factor equations reject after Holm adjustment under the inclusive $p\le .05$ rule: MKT, SMB, HML, CMA, and Momentum. RMW does not reject. Its selected observed scale remains $q=.27$, but its primary $p$-value rises from .050 under the shorter $q\le .50$ search to .062 because the null maximum is recomputed over the wider search domain. The result illustrates why radius search belongs inside every bootstrap draw.

Unemployment-rate change is the only macro equation that rejects after Holm adjustment under the primary grid ($p=.007$, Holm $p=.028$). This family-wise conclusion is not invariant to the tail grid. With the additional radii $.80,.85,.90$, unemployment remains individually significant ($p=.021$), but its macro-family Holm value becomes .084. We therefore describe the individual unemployment result as tail-grid robust and the macro-family conclusion as tail-grid sensitive.

The linear benchmark changes the substantive interpretation. Five contiguous time-ordered cross-fitting folds produce positive out-of-sample $R^2$ for most equations but negative values for RMW and the macro market-return equation. More importantly, cross-fitted residual MBCMI rejects none of the ten full-sample equations. The nuisance OLS model is refitted inside every bootstrap draw before residual Ball is recomputed. Full-sample MBCMI rejections therefore establish contemporaneous conditional-mean dependence in the unresidualized outcomes, not distinctive nonlinear or beyond-linear structure.

\subsection{Pointwise rolling evidence}

\begin{table}[!htbp]
\centering
\caption{Pointwise rolling results under the primary $q\le .75$ search}
\label{tab:rolling-q075}
\begin{threeparttable}
\footnotesize
\begin{tabular}{lrrr}
\toprule
Outcome & Windows with $p<.05$ & Changes vs. $q\le .50$ & Windows selecting $q=.75$ \\
\midrule
MKT & 310 & 0 & 0 \\
SMB & 70 & 0 & 0 \\
HML & 310 & 0 & 0 \\
RMW & 120 & 31 & 0 \\
CMA & 310 & 0 & 0 \\
Momentum & 56 & 6 & 35 \\
Market excess return & 4 & 0 & 0 \\
Industrial-production growth & 18 & 0 & 19 \\
Unemployment-rate change & 105 & 0 & 65 \\
Baa-spread change & 1 & 0 & 3 \\
\midrule
Total & 1,304 & 37 & 122 \\
\bottomrule
\end{tabular}
\begin{tablenotes}[flushleft]\footnotesize
\item Notes: Each outcome has 310 overlapping 240-month windows. Adjacent windows share 239 observations. Entries are pointwise $p$-value summaries, not independent discoveries or simultaneous episode inference. All 3,100 windows completed with zero fatal tasks or failed bootstrap draws. The comparison column counts 5\% classification changes relative to the historical $q\le .50$ Stage-E series.
\end{tablenotes}
\end{threeparttable}
\end{table}

All 3,100 rolling windows complete without fatal tasks or failed bootstrap draws. Relative to the shorter $q\le .50$ series, 37 pointwise 5\% classifications change. Momentum gains six rejecting windows, from 50 to 56, while RMW loses 31, from 151 to 120; the other eight outcome counts are unchanged. The primary grid selects its .75 upper boundary in 122 windows, concentrated in Momentum, industrial-production growth, unemployment-rate change, and a few Baa-spread windows. Scale localisation in those windows remains boundary-limited: the tail extension is a sensitivity check, not proof that no still-larger maximiser exists.

These are descriptive trajectories. Adjacent 240-month windows share 239 observations, so runs cannot be interpreted as independent discoveries. No simultaneous or episode-level family-wise inference is claimed. Window-specific standardisation also means that temporal changes in the plotted statistic combine changes in conditional means, predictor support, correlation, density, and scale.

\subsection{Selected diagnostic windows and localisation}

\begin{table}[!htbp]
\centering
\caption{Selected diagnostic windows and the cross-fitted linear benchmark}
\label{tab:selected-diagnostics}
\begin{threeparttable}
\small
\begin{tabular}{llrrrr}
\toprule
Outcome & Window end & MBCMI $p$ & MBCMI $q$ & Residual $p$ & Residual $q$ \\
\midrule
HML & Jul. 2013 & 0.022 & 0.42 & 0.993 & 0.39 \\
Momentum & Jul. 2022 & 0.019 & 0.69 & 0.030 & 0.43 \\
Industrial-production growth & Sep. 2008 & 0.062 & 0.29 & 0.894 & 0.21 \\
Market excess return & Feb. 2000 & 0.197 & 0.43 & 0.709 & 0.15 \\
Unemployment-rate change & Dec. 1999 & 0.009 & 0.48 & 0.290 & 0.53 \\
\bottomrule
\end{tabular}
\begin{tablenotes}[flushleft]\footnotesize
\item Notes: These windows were selected from earlier exploratory rolling evidence and are descriptive rather than selection-adjusted confirmatory discoveries. Residual MBCMI uses five contiguous cross-fitting folds and refits OLS inside every bootstrap draw. Momentum in July 2022 is the only selected window that retains a 5\% rejection after removing the cross-fitted contemporaneous linear component. Under the capped-support W4 weighting sensitivity, its residual $p$-value remains 0.030.
\end{tablenotes}
\end{threeparttable}
\end{table}

The five windows in Table~\ref{tab:selected-diagnostics} were selected from exploratory rolling calculations and do not receive selection-adjusted inference. HML in July 2013 and unemployment-rate change in December 1999 reject in the unresidualized MBCMI calculation but not after the cross-fitted linear benchmark. Momentum in July 2022 is the sole selected window that retains a residual rejection: its unresidualized MBCMI test selects $q=.69$ with $p=.019$, while residual MBCMI selects $q=.43$ with $p=.030$.

Localisation is interpreted at the level of broad high-contribution unions rather than exact top-ten centre rankings. That choice follows the radius and studentisation audits: exact leading-centre identities can be sensitive even when broad state regions remain stable. The selected radius is likewise a descriptive maximiser over the finite search domain, not an estimate of structural signal width.

Support weighting does not drive the surviving Momentum diagnostic. Capping the support weights (W4) changes no full-sample, diagnostic, or Holm conclusion, and the July-2022 residual Momentum $p$-value remains .030 under both W1 and W4. Finance $k$NN $p$-values are not used for formal inference because the matched time-series implementation failed its prespecified size criterion.

\section{Discussion}
\label{sec:discussion}

The evidence supports a bounded methodological claim. MBCMI is useful when conditional-mean departures may occupy spatially coherent regions of a multivariate predictor space at an unknown scale and when localisation matters alongside rejection. MBCMI's strength does not come from uniformly higher average power. It comes from a geometry: common spatial extent, variable local mass, multiscale search, and an observed-centre decomposition.

\subsection{Interpretation}

The extended comparator study clarifies the geometry over which MBCMI is strongest. NCMD is much stronger for interaction and checkerboard alternatives, where rapid sign changes align poorly with fixed-radius averaging. MBCMI is stronger for local islands, annular signals, and local absolute departures, and it is competitive in several weak linear settings. $k$NN remains a strong overall benchmark. The appropriate conclusion is a power envelope rather than a winner.

The predictor-law experiment reaches the same type of conclusion about metric choice. Ordinary z-scoring and robust-MAD scaling are close in many designs; whitening can be valuable for clustered support but harmful when correlation contains signal-relevant geometry. Standardisation is therefore part of the representation, not a neutral preprocessing step.

Finance results further restrict interpretation. Full-sample MBCMI rejections largely disappear after removing a cross-fitted contemporaneous linear component. Those results should be read as evidence that current factor or macro states are associated with the current conditional mean, not as proof of nonlinear mechanisms, forecasting ability, or causality. The selected July-2022 Momentum case is different: the residual test remains significant and survives the capped-weight sensitivity, so conditional-mean structure remains after the cross-fitted linear benchmark in that selected window. Because the window was selected from exploratory rolling evidence, even that result remains descriptive rather than selection-adjusted confirmation.

\subsection{Limitations}

Formal validity now covers two deliberately bounded classes. Primitive iid theory proves the Ball-kernel reduction and raw Rademacher maximum on a fixed grid under compact regular support, uniform local mass, and moment conditions. A separate serial theorem covers the implemented recursive-sign algorithm for stable finite-order autoregressions with conditionally sign-symmetric innovations, including BIC order selection and least-squares innovation replacement. Neither result is model-free. Fixed-grid consistency is only grid-visible; the iid support assumptions exclude the unbounded predictor laws used in several robustness experiments; and the serial theorem does not cover general asymmetric martingale-difference innovations, nonlinear dynamics, or long memory. Simulations probe some of those departures without converting them into theorem-covered cases.

Computational robustness is broad but not universal. Studentisation is mildly conservative in the serial null and materially changes selected scales. Predictor geometry changes power under correlation and clustering. MDD/MDC remains an important estimand-matched literature benchmark. The tested public author-code MDD implementation was operational, but it failed the prespecified finite-sample calibration screen and therefore did not enter the power tournament; that exclusion is an implementation-comparability decision, not a claim that MDD is theoretically invalid or inferior. NCMD $K=5$ and $K=10$ are both reported because selecting $K$ after viewing power would compromise the comparison.

Rolling windows overlap heavily and receive no simultaneous multiplicity correction. Their episode structure is descriptive. Moreover, 122 rolling windows select the primary upper boundary, so their selected scales should not be interpreted as fully resolved spatial widths. The $q=.90$ tail sensitivity also shows that a family-wise macro conclusion can depend on the search domain even when the individual unemployment rejection survives. These qualifications narrow the claims but make the final procedure and evidence internally coherent.

\section{Conclusion}
\label{sec:conclusion}

We develop a Multiscale Ball Conditional Mean Independence (MBCMI) test that searches fixed-radius neighbourhoods while retaining a centre-by-radius localisation. The final procedure uses a raw support-weighted maximum over 71 empirical distance-quantile radii and repeats the complete search inside every resampling draw. Quantiles cover from 0.05 to 0.75 of the longest distance across the dataset, providing an extensive search range on any dataset. Fixed-grid theory supplies the population target and consistency result; targeted iid and serial null studies calibrate the implemented procedure within the stated empirical scope.

The computational evidence defines where the method helps. MBCMI is strong for several spatially coherent local and radial alternatives, whereas nearest-neighbour conditional-mean methods dominate important sign-changing designs. Predictor-distribution robustness is broad but geometry-specific, and studentisation does not improve the serial procedure sufficiently to replace the raw statistic. These results support multiscale adaptation and localisation, not universal power dominance.

The finance applications illustrate the distinction between detection and interpretation. Five factor equations and unemployment-rate change reject in the primary full-sample analysis, but none of the ten full-sample cross-fitted residual tests rejects. Full-sample evidence is therefore contemporaneous conditional-mean dependence rather than distinctive beyond-linear structure. A selected July-2022 Momentum window is the one case that remains significant after the cross-fitted linear benchmark and capped-weight sensitivity. Rolling trajectories remain pointwise descriptions of overlapping samples.

Theory now matches the two calibration mechanisms at different levels of generality. For iid data, the implemented raw finite-grid maximum has primitive ball-kernel validity and an explicit Pitman local-power limit. For serial data, the implemented recursive-sign algorithm is valid for a stable finite-order, conditionally sign-symmetric autoregressive class after BIC and least-squares replacement. Broader asymmetric martingale-difference dynamics remain outside that theorem and are evaluated by targeted serial calibration experiments. The resulting contribution of MBCMI is therefore deliberately finite-grid and assumption-specific: a conditional-mean test with a formal iid core, a restricted feasible serial extension, a geometry-specific power envelope, and an empirical application whose claims remain bounded by those scopes.

\section*{Data and Code Availability}

All factor and macroeconomic source series are public and are identified in Table~\ref{tab:finance-design}; Appendix~\ref{app:finance-robustness} records exact source identifiers and transformations. The simulation comparisons use deterministic seed maps and matched datasets where paired inference is reported. A separate research archive documents the configurations, summary outputs, seed records, and audit manifests used to construct the paper's exhibits. The arXiv source package contains the manuscript source and displayed exhibit files rather than the complete executable research repository. No proprietary data enter the analysis.

\bibliographystyle{apalike}
\bibliography{references}

\clearpage
\appendix
\begin{center}
{\Large\bfseries Technical Appendix and Supplementary Evidence}\par
\vspace{0.4em}
{\small Formal proofs, exact DGP definitions, extended simulation evidence, finance calibration, and reproducibility records}\par
\end{center}
\clearpage
\section*{Appendix Guide}
\label{app:guide}

The technical appendix separates proof detail from extended computational evidence. Appendices~\ref{app:population}--\ref{app:proof-audit} develop the population target, fixed-grid asymptotics, primitive iid raw-maximum calibration, the recursive-sign serial result, and a formal evidence audit. Appendix~\ref{app:implementation} records the implemented procedure. Appendix~\ref{app:canonical-extension} reports the canonical Ball--NCMD comparison. Appendix~\ref{app:predictor-robustness} records predictor-law and metric robustness. Appendix~\ref{app:finance-robustness} reports data construction, linear/residual benchmarks, weighting, and rolling evidence. Appendix~\ref{app:repro} records computational provenance.

\section{Notation and Population Target}
\label{app:population}

Let $(Y_i,X_i)_{i=1}^n$ be observations with $Y_i\in\R$ and
$X_i\in\R^d$. Write
\[
m(x)=\E[Y\mid X=x],
\qquad
\mu=\E[Y].
\]
We test mean independence of $Y$ from $X$,
\begin{equation}
H_0^{CM}:\quad m(X)=\mu\quad\text{almost surely}.
\label{eq:app-null}
\end{equation}
Full independence implies \eqref{eq:app-null}, but the converse need not
hold. Conditional variances, tails, and higher moments may depend on $X$ while
$m(X)$ remains constant.

Let $a=\E[X]$ and let $D$ be the diagonal matrix of positive marginal
standard deviations. Population-standardised predictors are
$Z=D^{-1}(X-a)$. For an independent copy $(Y',Z')$, define
\begin{align}
p_{e}(z)
&=\Pp\{\|Z'-z\|\le e\},\\
r_{e}(z)
&=\E[Y'\1\{\|Z'-z\|\le e\}],\\
m_{e}(z)
&=r_e(z)/p_e(z)
\end{align}
when $p_e(z)>0$. Our population criterion is
\begin{equation}
Q(e)
=
\E\left[p_e(Z)\{m_e(Z)-\mu\}^2\right].
\label{eq:app-Q}
\end{equation}
Probability weighting mirrors the sample neighbourhood-size weight.

\begin{proposition}[Null implication]
\label{prop:app-null}
If $\E|Y|<\infty$, then \eqref{eq:app-null} implies $Q(e)=0$ for every
radius with well-defined local means.
\end{proposition}

\begin{proof}
Under \eqref{eq:app-null}, iterated expectations give
\[
r_e(z)
=
\E[m(Z')\1\{\|Z'-z\|\le e\}]
=
\mu p_e(z).
\]
Thus $m_e(z)=\mu$ whenever $p_e(z)>0$, and the nonnegative integrand in
\eqref{eq:app-Q} vanishes.
\end{proof}

\begin{assumption}[Predictor support and local mass]
\label{ass:support}
Support $\mathcal Z$ of $Z$ is compact and satisfies a uniform interior
cone condition. Predictor $Z$ has a density bounded above and bounded away
from zero on $\mathcal Z$. Regression function $m$ is bounded and continuous
except on a set of predictor probability zero.
\end{assumption}

\begin{proposition}[Local identification]
\label{prop:app-identification}
Suppose Assumption~\ref{ass:support} holds and
$\Pp\{m(Z)\ne\mu\}>0$. Then there is $e_0>0$ such that
$Q(e)>0$ for every $e\in(0,e_0)$.
\end{proposition}

\begin{proof}
Choose $\delta>0$ such that
$A_\delta=\{z:|m(z)-\mu|>2\delta\}$ has positive probability. Almost
every point of $A_\delta$ is a continuity point of $m$ and a Lebesgue point
of the predictor density. For each such point, local averaging gives
$m_e(z)\to m(z)$ as $e\downarrow0$. Egorov's theorem supplies a
positive-probability subset on which this convergence is uniform. Hence
$|m_e(z)-\mu|>\delta$ on that subset for all sufficiently small $e$.
Assumption~\ref{ass:support} gives $p_e(z)>0$, so integration in
\eqref{eq:app-Q} is strictly positive.
\end{proof}

\begin{remark}[Estimand]
A variance-only process $Y=\mu+\sigma(X)U$ with
$\E[U\mid X]=0$ satisfies $Q(e)=0$ at every radius. Such processes are null
size designs, not power alternatives.
\end{remark}

\section{Sample Statistic and Fixed-Grid Asymptotics}
\label{app:sample}

Implementation replaces $a$ and $D$ by their sample analogues and writes
$\widehat Z_i=\widehat D^{-1}(X_i-\bar X)$. For centre $i$ and radius $e$,
define
\begin{align}
B_i(e)&=\{j:\|\widehat Z_j-\widehat Z_i\|\le e\},\\
N_i(e)&=|B_i(e)|,\\
\bar Y_i(e)&=N_i(e)^{-1}\sum_{j\in B_i(e)}Y_j.
\end{align}
Let $\mathcal I_n(e)$ contain centres satisfying the fixed minimum-size
rule. A fixed-radius statistic is
\begin{equation}
T_n(e)
=
\sum_{i\in\mathcal I_n(e)}N_i(e)
\{\bar Y_i(e)-\bar Y\}^2.
\label{eq:app-T}
\end{equation}
Our radius grid contains the 71 pairwise-distance quantiles
$q\in\mathcal Q=\{0.05,0.06,\ldots,0.75\}$. A candidate remains
admissible when at least 20\% of centres satisfy the minimum-size rule. The
reported statistic is
\begin{equation}
T_n^{\max}=\max_{q\in\widehat{\mathcal Q}_n}
T_n(\widehat e_{n,q}).
\label{eq:app-Tmax}
\end{equation}
A paired 20\%--40\% sensitivity audit changed no selected radius, observed
maximum, resampling $p$-value, or rejection decision in 13,000 Monte Carlo
comparisons or in the four principal empirical windows.

\subsection{A leave-one-out representation}

Self-inclusion complicates notation but not the first-order limit. Define
\[
\widehat p_{i,e}^{(-i)}
=
\frac{1}{n-1}\sum_{j\ne i}
\1\{\|\widehat Z_j-\widehat Z_i\|\le e\},
\quad
\widehat r_{i,e}^{(-i)}
=
\frac{1}{n-1}\sum_{j\ne i}Y_j
\1\{\|\widehat Z_j-\widehat Z_i\|\le e\},
\]
and $\widehat m_{i,e}^{(-i)}=\widehat r_{i,e}^{(-i)}/
\widehat p_{i,e}^{(-i)}$. Let
\begin{equation}
\widetilde T_n(e)
=
(n-1)\sum_{i=1}^n
\widehat p_{i,e}^{(-i)}
\{\widehat m_{i,e}^{(-i)}-\bar Y\}^2.
\label{eq:loo-stat}
\end{equation}
For fixed positive radii, adding the centre changes each local numerator and
denominator by $O(n^{-1})$. Consequently
$|T_n(e)-\widetilde T_n(e)|=o_p(n^2)$ uniformly over a fixed finite grid.

\begin{assumption}[IID fixed-grid regularity]
\label{ass:iid-grid}
Observations are iid; $\E|Y|^{4+\eta}<\infty$ for some $\eta>0$;
Assumption~\ref{ass:support} holds; every population radius considered below
is positive; and the fixed minimum-neighbour count is $o(n)$. Marginal
standard deviations of $X$ are bounded away from zero. For every indexed
radius $e_q$, there are constants $p_0>0$, $C<\infty$, and $t_0>0$ such that
\[
\inf_{z\in\mathcal Z}p_{e_q}(z)\ge p_0,
\qquad
\sup_{z\in\mathcal Z}
\Pp\{|\|Z-z\|-e_q|\le t\}\le Ct,
\quad 0<t<t_0.
\]
The pairwise-distance distribution $F_D$ is continuously differentiable on
a neighbourhood of every $e_q$, with $0<f_D(e_q)<\infty$. The quantile set
$\mathcal Q$ is finite and fixed as $n\to\infty$.
\end{assumption}

\begin{lemma}[Uniform local-mean convergence]
\label{lem:uniform-local}
Under Assumption~\ref{ass:iid-grid}, for every fixed finite radius set
$\mathcal E$,
\[
\max_{e\in\mathcal E}\max_{1\le i\le n}
|\widehat p_{i,e}^{(-i)}-p_e(Z_i)|=o_p(1)
\]
and
\[
\max_{e\in\mathcal E}
\frac1n\sum_{i=1}^n
|\widehat m_{i,e}^{(-i)}-m_e(Z_i)|^2=o_p(1).
\]
These conclusions also hold after replacing population standardisation by
sample standardisation.
\end{lemma}

\begin{proof}
Indicators of Euclidean balls form a VC class. Multiplying by an envelope
with finite $2+\eta$ moment preserves the required Glivenko--Cantelli
property for the local numerators. Uniform positive local mass makes the
ratio map Lipschitz on the relevant event. Compact support and consistency of
$\bar X$ and $\widehat D$ make sample-standardised distances uniformly
close to population-standardised distances. Boundary events are negligible
because the distance distribution is continuous at each fixed radius.
\end{proof}

\begin{theorem}[Fixed-radius probability limit]
\label{thm:fixed-plim}
Under Assumption~\ref{ass:iid-grid}, for every fixed $e\in\mathcal E$,
\[
n^{-2}T_n(e)\xrightarrow{p}Q(e).
\]
Convergence holds jointly over every fixed finite set $\mathcal E$.
\end{theorem}

\begin{proof}
A leave-one-out representation gives
\[
n^{-2}\widetilde T_n(e)
=
\frac{n-1}{n}\frac1n\sum_{i=1}^n
\widehat p_{i,e}^{(-i)}
\{\widehat m_{i,e}^{(-i)}-\bar Y\}^2.
\]
Lemma~\ref{lem:uniform-local}, $\bar Y\to_p\mu$, and the moment
condition imply that the displayed average differs by $o_p(1)$ from
$n^{-1}\sum_i p_e(Z_i)\{m_e(Z_i)-\mu\}^2$. An ordinary law of large
numbers gives $Q(e)$. Asymptotic equivalence of $T_n(e)$ and
$\widetilde T_n(e)$ completes the proof. Finiteness of $\mathcal E$ gives
joint convergence.
\end{proof}

\begin{corollary}[Empirical distance quantiles]
\label{cor:emp-quantiles}
Let $e_q$ be the population $q$th quantile of $\|Z-Z'\|$ and let
$\widehat e_{n,q}$ be its empirical pairwise-distance analogue. Under
Assumption~\ref{ass:iid-grid}, for every fixed finite
$\mathcal Q$,
\[
\max_{q\in\mathcal Q}|\widehat e_{n,q}-e_q|=o_p(1)
\]
and
\[
\max_{q\in\mathcal Q}
|n^{-2}T_n(\widehat e_{n,q})-Q(e_q)|=o_p(1).
\]
Coverage screening is asymptotically inactive for these radii.
\end{corollary}

\begin{proof}
An empirical pairwise-distance distribution is a U-empirical distribution
and converges uniformly to its population counterpart. Positive density at
each indexed quantile gives quantile consistency. Stochastic equicontinuity
in the radius follows from the VC property of balls and zero probability on
sphere boundaries. Uniform positive local mass implies that every centre
satisfies the fixed minimum-size rule with probability approaching one, so
both the 20\% screen and any stricter fixed screen admit the same grid
asymptotically.
\end{proof}

\begin{theorem}[Consistency for grid-visible fixed alternatives]
\label{thm:fixed-consistency}
Suppose Assumption~\ref{ass:iid-grid} holds and
$Q(e_q)>0$ for at least one $q\in\mathcal Q$. Then
\[
n^{-2}T_n^{\max}\xrightarrow{p}
\max_{q\in\mathcal Q}Q(e_q)>0.
\]
If a null critical value is $O_p(n)$, the resulting test is consistent
against this fixed, grid-visible alternative.
\end{theorem}

\begin{proof}
Corollary~\ref{cor:emp-quantiles} and continuity of the maximum over a fixed
finite vector give the first statement. Under the null, the local means have
sampling variation of order $N_i(e)^{-1/2}$, so each weighted squared
contrast is $O_p(1)$ and the finite-grid maximum is $O_p(n)$. Under the stated grid-visible alternative, $T_n^{\max}$ is of order $n^2$.
Local identification alone does not guarantee grid visibility because every
population distance quantile in the fixed index set remains positive.
\end{proof}

\subsection{Pitman local alternatives}

The fixed-grid criterion yields a direct local-power representation. For
$q\in\mathcal Q$, define the Hilbert space
$\mathbb H_q=L^2(P_C)$ and the feature map
\begin{equation}
\Psi_q(c,z)
=
\sqrt{p_{e_q}(c)}
\left\{
\frac{\1\{\|z-c\|\le e_q\}}{p_{e_q}(c)}-1
\right\}.
\label{eq:app-feature-map}
\end{equation}
Then
$K_{e_q}(z,z')=\langle\Psi_q(\cdot,z),
\Psi_q(\cdot,z')\rangle_{\mathbb H_q}$.

Consider a triangular array
\begin{equation}
Y_{n,i}=\mu+n^{-1/2}\delta_0(Z_i)+\varepsilon_i,
\qquad
\E\{\delta_0(Z)\}=0,
\qquad
\E(\varepsilon_i\mid Z_i)=0,
\label{eq:app-local-alt}
\end{equation}
where $\delta_0$ is bounded and square-integrable and the conditional moments
of $\varepsilon_i$ satisfy Assumption~\ref{ass:iid-null}. Set
\begin{equation}
\Delta_q(c)
=
\E\{\Psi_q(c,Z)\delta_0(Z)\}.
\label{eq:app-local-shift}
\end{equation}
Because $\E\delta_0(Z)=0$,
\[
\Delta_q(c)
=
\sqrt{p_{e_q}(c)}
\E\{\delta_0(Z)\mid \|Z-c\|\le e_q\},
\]
and therefore
\begin{equation}
\|\Delta_q\|_{\mathbb H_q}^2
=
\E_C\left[
 p_{e_q}(C)
 \left\{
 \E[\delta_0(Z)\mid\|Z-C\|\le e_q]
 \right\}^2
\right]
=:Q_\delta(e_q).
\label{eq:app-local-noncentrality}
\end{equation}
Equation~\eqref{eq:app-local-noncentrality} is the population Ball criterion
applied to the local mean departure. It makes cancellation explicit: a
rapidly sign-changing $\delta_0$ can have small $Q_\delta(e_q)$ at a radius
that averages positive and negative pieces together.

\begin{proposition}[Pitman local-power limit]
\label{prop:app-local-power}
Under Assumptions~\ref{ass:iid-grid} and \ref{ass:iid-null}, with
\eqref{eq:app-local-alt}, jointly over $q\in\mathcal Q$,
\[
\left\{n^{-1}T_n(\widehat e_{n,q})\right\}_{q\in\mathcal Q}
\Rightarrow
\left\{\|\mathbb G_q+\Delta_q\|_{\mathbb H_q}^2\right\}_{q\in\mathcal Q},
\]
where $\{\mathbb G_q\}$ is the centred Gaussian Hilbert-space vector obtained
under the null. Conditional on the data, the Rademacher radius vector under
\eqref{eq:app-local-alt} converges to
$\{\|\mathbb G_q\|^2\}_{q\in\mathcal Q}$, so the ideal bootstrap critical
value converges to the null critical value $c_\alpha$. If at least one
$\Delta_q\ne0$, then
\[
\Pp\left\{
\max_q\|\mathbb G_q+\Delta_q\|^2>c_\alpha
\right\}>\alpha
\]
whenever the joint Gaussian law is nondegenerate on the closed linear span of
the shifts.
\end{proposition}

\begin{proof}
Lemma~\ref{lem:ball-matrix} below permits replacement of the sample statistic
by its population-kernel form uniformly over the finite grid. Using the
feature representation of $K_{e_q}$, the leading term can be written as
\[
\left\|
 n^{-1/2}\sum_{i=1}^n
 \Psi_q(\cdot,Z_i)
 \{\varepsilon_i+n^{-1/2}\delta_0(Z_i)\}
\right\|_{\mathbb H_q}^2+o_p(1).
\]
The first summand inside the norm converges jointly over the finite radius set
to a centred Gaussian element by the Hilbert-space central limit theorem;
the conditional $(4+\eta)$ moment bound and bounded feature envelope imply
the required Lindeberg and tightness conditions. The second summand equals
$n^{-1}\sum_i\Psi_q(\cdot,Z_i)\delta_0(Z_i)$ and converges in
$\mathbb H_q$ to $\Delta_q$ by the Hilbert-space law of large numbers.
Continuous mapping gives the displayed joint limit.

For the bootstrap critical value, write the centred residual as
$\widehat u_{n,i}=\varepsilon_i+n^{-1/2}\{\delta_0(Z_i)-\overline\delta_n\}-\overline{\varepsilon}_n$. Conditional Rademacher multiplication removes the deterministic local mean shift at first order: the extra multiplier sum has conditional second moment of order
\[
\frac{1}{n^2}\sum_{i=1}^n\|\Psi_q(\cdot,Z_i)\|^2
\{\delta_0(Z_i)-\overline\delta_n\}^2=O_p(n^{-1}),
\]
uniformly over the finite grid. The residual-mean term is smaller by the same argument. Hence the conditional multiplier process has the same centred Gaussian limit as under the null, and the ideal raw-max bootstrap critical value converges to $c_\alpha$.

For the power statement, the acceptance set
$\{(g_q)_q:\max_q\|g_q\|^2\le c_\alpha\}$ is closed, convex, and symmetric in
the product Hilbert space. Anderson's inequality for Gaussian measures
implies that translating a nondegenerate centred Gaussian law by a nonzero
vector weakly lowers the probability of this set, with strict inequality
under the stated nondegeneracy and continuity conditions. Taking complements
gives the result.
\end{proof}

\begin{remark}[Scope of local power]
Proposition~\ref{prop:app-local-power} is a Pitman result for the fixed
71-radius family. It does not establish minimax optimality, adaptation to a
growing grid, or detection of departures whose Ball-smoothed projections
vanish at every indexed radius.
\end{remark}

\section{IID Calibration}
\label{app:calibration}

\subsection{Exchangeability benchmark}

For permutation $\pi$, pair $X_i$ with $Y_{\pi(i)}$ and recompute the
complete maximum. Its Monte Carlo $p$-value is
\[
\widehat p_{\rm perm}
=
\frac{1+\sum_{b=1}^B
\1\{T_{n,b}^{\max}\ge T_n^{\max}\}}{B+1}.
\]

\begin{proposition}[Finite-sample permutation validity]
\label{prop:app-perm}
Conditional on $X_1,\ldots,X_n$, suppose the joint distribution of
$(Y_1,\ldots,Y_n)$ is invariant to every permutation under the null. Then
the randomisation test has conditional size no greater than $\alpha$.
\end{proposition}

\begin{proof}
Permutation invariance makes the observed and permuted maxima exchangeable
conditional on $X$. Rank uniformity gives the result. Recomputing the complete
search inside every permutation incorporates radius selection.
\end{proof}

Mean independence alone does not imply this invariance. Unrestricted
permutation therefore remains an exchangeability benchmark rather than the
primary calibration for \eqref{eq:app-null}.

\subsection{Rademacher wild-multiplier calibration}

Write $\widehat u_i=Y_i-\bar Y$ and generate
\begin{equation}
Y_i^*=\bar Y+\xi_i\widehat u_i,
\qquad
\Pp(\xi_i=1)=\Pp(\xi_i=-1)=1/2.
\label{eq:app-wild}
\end{equation}
Geometry depends only on $X$, so every multiplier draw uses the observed
standardisation and empirical distance-quantile radii. The implementation
re-evaluates coverage and the full maximum, although coverage is unchanged
across outcome multipliers when geometry is fixed.

\subsubsection{Population Ball kernel}

For a population-standardised centre $c$ and predictor value $z$, define
\begin{equation}
g_e(c,z)
=
\frac{\1\{\|z-c\|\le e\}}{p_e(c)}-1.
\label{eq:app-ge}
\end{equation}
Let $C$ be an independent draw from the predictor distribution and set
\begin{equation}
K_e(z,z')
=
\E\left[
p_e(C)g_e(C,z)g_e(C,z')
\right].
\label{eq:app-Ke}
\end{equation}

\begin{lemma}[Ball-kernel properties]
\label{lem:ball-kernel}
Under Assumption~\ref{ass:iid-grid}, $K_e$ is bounded, symmetric, and
positive semidefinite for every indexed radius. In particular, for every
square-integrable $a$,
\[
\E\{a(Z)K_e(Z,Z')a(Z')\}
=
\E_C\left[
p_e(C)
\left\{
\E[a(Z)g_e(C,Z)]
\right\}^2
\right]
\ge0.
\]
\end{lemma}

\begin{proof}
Uniform local mass gives $p_e(C)\ge p_0>0$, so
$|g_e(C,Z)|\le1+p_0^{-1}$ and $K_e$ is bounded. Symmetry is immediate.
Fubini's theorem gives the displayed nonnegative quadratic form.
\end{proof}

For an admissible sample radius, let $H_{n,e}$ have row $i$
\[
H_{n,e,ij}
=
\frac{\1\{\|\widehat Z_j-\widehat Z_i\|\le e\}}{N_i(e)}
\]
and define $A_{n,e}=H_{n,e}-\mathbf1\mathbf1'/n$ and
$M_{n,e}=A_{n,e}'D_{n,e}A_{n,e}$, with
$D_{n,e}=\operatorname{diag}\{N_i(e)\}$. Rows corresponding to inadmissible
centres are omitted. Then $T_n(e)=Y'M_{n,e}Y$ and
\begin{equation}
M_{n,e}\mathbf1=0.
\label{eq:app-annihilate}
\end{equation}
Equation~\eqref{eq:app-annihilate} is exact. Hence
$T_n(e)=u'M_{n,e}u=\widehat u'M_{n,e}\widehat u$ for the observed statistic;
estimating the global mean does not create an observed-statistic
replacement error.

\begin{lemma}[Uniform denominator, quantile, and shell control]
\label{lem:ball-shell}
Under Assumption~\ref{ass:iid-grid}, uniformly over $q\in\mathcal Q$,
\begin{align}
\max_{1\le i\le n}
\left|\frac{N_i(e_q)}{n}-p_{e_q}(Z_i)\right|
&=O_p\!\left(\sqrt{\frac{\log n}{n}}\right),
\label{eq:app-denom-rate}\\
\widehat e_{n,q}-e_q&=O_p(n^{-1/2}),
\label{eq:app-uq-rate}\\
\max_{r,s}
\left|
\|\widehat Z_r-\widehat Z_s\|-\|Z_r-Z_s\|
\right|&=O_p(n^{-1/2}).
\label{eq:app-standard-rate}
\end{align}
If $\rho_n=O_p(n^{-1/2})$, then
\begin{equation}
\frac1{n^2}\sum_{r,s=1}^n
\1\big\{|\|Z_r-Z_s\|-e_q|\le |\rho_n|\big\}
=o_p(1)
\label{eq:app-shell-rate}
\end{equation}
uniformly over the fixed quantile set. Consequently the proportion of
pairwise memberships changed by sample standardisation and empirical-radius
replacement is $o_p(1)$.
\end{lemma}

\begin{proof}
For fixed $e_q$, the class
$\{z\mapsto\1(\|z-c\|\le e_q):c\in\mathcal Z\}$ is VC with a bounded
envelope. A standard VC maximal inequality gives a uniform empirical-process
bound of order $O_p\{\sqrt{(\log n)/n}\}$. Evaluating the class at the random
sample centres and adding the $1/n$ self-inclusion term gives
\eqref{eq:app-denom-rate}. Finiteness of $\mathcal Q$ makes the bound
simultaneous over radii.

The empirical pairwise-distance distribution is a U-empirical distribution
with bounded kernel. Its U-quantile admits the ordinary root-$n$ Bahadur
rate at an interior quantile where $f_D(e_q)$ exists and is strictly
positive; see, for example, the U-statistic quantile results collected in
\citet{serfling1980approximation}. This yields \eqref{eq:app-uq-rate}.

Compact support and marginal standard deviations bounded away from zero give
$\bar X-a=O_p(n^{-1/2})$ and $\widehat D-D=O_p(n^{-1/2})$. The map
$(x,x',a,D)\mapsto\|D^{-1}(x-a)-D^{-1}(x'-a)\|$ is uniformly Lipschitz on
the resulting compact parameter set, proving \eqref{eq:app-standard-rate}.

For \eqref{eq:app-shell-rate}, choose a deterministic $b_n\downarrow0$ with
$\Pp(|\rho_n|>b_n)\to0$ and $b_n=O(n^{-1/2}\log n)$. On
$\{|\rho_n|\le b_n\}$, the left side is bounded by the U-empirical average
of the shell indicator with width $b_n$. Assumption~\ref{ass:iid-grid}
gives shell probability at most $Cb_n$, and the bounded U-statistic law of
large numbers makes the empirical average $Cb_n+o_p(1)=o_p(1)$. A union
bound over the finite grid completes the proof. Combining
\eqref{eq:app-uq-rate}, \eqref{eq:app-standard-rate}, and the triangle
inequality shows that every changed membership lies in such a shrinking
shell.
\end{proof}

\begin{lemma}[Primitive Ball-matrix bounds and kernel reduction]
\label{lem:ball-matrix}
Suppose Assumption~\ref{ass:iid-grid} holds and the fixed quantile set is
$\mathcal Q$. For $M_{n,q}=M_{n,\widehat e_{n,q}}$, uniformly over
$q\in\mathcal Q$,
\begin{align}
\max_{r,s}|M_{n,q,rs}|&=O_p(1),\label{eq:matrix-entry}\\
\max_r\sum_s M_{n,q,rs}^2&=O_p(n),\label{eq:matrix-row}\\
\sum_r M_{n,q,rr}^2&=O_p(n).\label{eq:matrix-diag}
\end{align}
Moreover,
\begin{equation}
\max_{q\in\mathcal Q}
\frac1{n^2}
\sum_{r,s=1}^n
\left[
M_{n,q,rs}-K_{e_q}(Z_r,Z_s)
\right]^2
=o_p(1),
\label{eq:matrix-L2}
\end{equation}
and the corresponding average diagonal difference converges to zero.
Under $H_0^{CM}$ and the moment conditions in
Assumption~\ref{ass:iid-null} below,
\begin{equation}
\max_{q\in\mathcal Q}
\left|
\frac{T_n(\widehat e_{n,q})}{n}
-
\frac1n\sum_{r=1}^n\sum_{s=1}^n
K_{e_q}(Z_r,Z_s)u_ru_s
\right|
=o_p(1).
\label{eq:kernel-reduction}
\end{equation}
The same reduction holds conditionally for the Rademacher bootstrap.
\end{lemma}

\begin{proof}
We prove the result in five steps. Because $\mathcal Q$ is finite, every
bound below may be intersected over $q$ without changing its stochastic
order.

\emph{Step 1: exact matrix representation and entry bounds.}
On the event that every centre is admissible, write
$I_{ir}(e)=\1\{\|\widehat Z_r-\widehat Z_i\|\le e\}$ and
$\widehat p_i(e)=N_i(e)/n$. Direct expansion of
$M_{n,e}=A_{n,e}'D_{n,e}A_{n,e}$ gives
\begin{equation}
M_{n,e,rs}
=
\frac1n\sum_{i=1}^n
\left\{
\frac{I_{ir}(e)I_{is}(e)}{\widehat p_i(e)}
-I_{ir}(e)-I_{is}(e)+\widehat p_i(e)
\right\}.
\label{eq:app-m-entry-exact}
\end{equation}
Lemma~\ref{lem:ball-shell} and uniform positive local mass imply
$\inf_i\widehat p_i(e_q)\ge p_0/2$ with probability approaching one.
Every summand in \eqref{eq:app-m-entry-exact} is therefore bounded by a
constant independent of $n$, which proves \eqref{eq:matrix-entry}. The same
bound gives
$\sum_sM_{n,q,rs}^2\le Cn$ and
$\sum_rM_{n,q,rr}^2\le Cn$, proving
\eqref{eq:matrix-row}--\eqref{eq:matrix-diag}. Assumption
\ref{ass:iid-grid} also implies that all centres satisfy the fixed
minimum-neighbour rule with probability approaching one, so omitting
inadmissible rows does not affect these conclusions.

\emph{Step 2: replace empirical denominators at population geometry.}
First hold the population-standardised observations and population radius
$e_q$ fixed. Define
\[
\overline K_{n,q}(z,z')
=
\frac1n\sum_{i=1}^n
p_{e_q}(Z_i)g_{e_q}(Z_i,z)g_{e_q}(Z_i,z').
\]
Expanding the product shows that the only difference between the exact
matrix entry in \eqref{eq:app-m-entry-exact} and
$\overline K_{n,q}(Z_r,Z_s)$ comes from replacing
$\widehat p_i(e_q)$ by $p_{e_q}(Z_i)$. On the event
$\inf_i\widehat p_i\wedge p_i\ge p_0/2$,
\[
\left|\widehat p_i^{-1}-p_i^{-1}\right|
\le \frac{2}{p_0^2}|\widehat p_i-p_i|.
\]
Hence Lemma~\ref{lem:ball-shell} implies
\[
\max_{r,s}
|M_{n,e_q,rs}^{(0)}-\overline K_{n,q}(Z_r,Z_s)|
=o_p(1),
\]
where $M^{(0)}$ denotes the matrix constructed with population
standardisation and radius. In particular its empirical squared average is
$o_p(1)$.

\emph{Step 3: replace the empirical centre average by the population kernel.}
Let
$f_q(c,z,z')=p_{e_q}(c)g_{e_q}(c,z)g_{e_q}(c,z')$. Uniform local mass makes
$f_q$ bounded. Expanding the empirical squared error gives
\begin{align*}
&\E\left[
\frac1{n^2}\sum_{r,s}
\{\overline K_{n,q}(Z_r,Z_s)-K_{e_q}(Z_r,Z_s)\}^2
\right]\\
&\qquad=
\E\left[
\left\{
\frac1n\sum_i f_q(Z_i,Z_1,Z_2)
-\E_C f_q(C,Z_1,Z_2)
\right\}^2
\right]+O(n^{-1}).
\end{align*}
The $O(n^{-1})$ term collects configurations in which the centre index
coincides with one of the evaluation indices. Conditional on $(Z_1,Z_2)$,
the remaining centre terms are iid and bounded, so their variance is
$O(n^{-1})$. The displayed expectation is therefore $O(n^{-1})$.
Markov's inequality yields
\[
\frac1{n^2}\sum_{r,s}
\{\overline K_{n,q}(Z_r,Z_s)-K_{e_q}(Z_r,Z_s)\}^2=o_p(1).
\]
The same argument with $r=s$ gives
$n^{-1}\sum_r|\overline K_{n,q}(Z_r,Z_r)-K_{e_q}(Z_r,Z_r)|=o_p(1)$.

\emph{Step 4: sample standardisation and empirical radii.}
Construct $M_{n,q}$ from $(\widehat Z_i,\widehat e_{n,q})$ and
$M_{n,q}^{(0)}$ from $(Z_i,e_q)$. Let
\[
\rho_{n,q}=\frac1{n^2}\sum_{i,j}
\left|I_{ij}(\widehat e_{n,q};\widehat Z)-I_{ij}(e_q;Z)\right|.
\]
Lemma~\ref{lem:ball-shell} gives $\rho_{n,q}=o_p(1)$. On the event on which
all neighbourhood probabilities are bounded below, formula
\eqref{eq:app-m-entry-exact} and
$|ab-a'b'|\le |a-a'|+|b-b'|$ for binary $a,b,a',b'$ give a constant $C$
independent of $n,r,s$ such that
\[
|M_{n,q,rs}-M_{n,q,rs}^{(0)}|
\le
\frac{C}{n}\sum_i
\{ |\Delta I_{ir}|+|\Delta I_{is}|+|\Delta \widehat p_i|\},
\]
where $\Delta I_{ij}$ is the membership change and
$\Delta\widehat p_i=n^{-1}\sum_j\Delta I_{ij}$. Jensen's inequality and
averaging first over $(r,s)$ and then over centres yield
\[
\frac1{n^2}\sum_{r,s}
\{M_{n,q,rs}-M_{n,q,rs}^{(0)}\}^2
\le C\rho_{n,q}=o_p(1).
\]
The diagonal analogue follows from the same bound. Combining Steps 2--4
proves \eqref{eq:matrix-L2} and the asserted average diagonal convergence.
The minimum-neighbour and 20\% coverage screens are inactive with
probability approaching one because
$\min_iN_i(\widehat e_{n,q})/n\ge p_0/3$ on the same event.

\emph{Step 5: transfer the matrix approximation to quadratic forms.}
Set
$\Delta_{n,q}=M_{n,q}-K_{n,q}$, where
$K_{n,q,rs}=K_{e_q}(Z_r,Z_s)$. Conditional on $Z_1,\ldots,Z_n$ and under
$H_0^{CM}$, independence and $\E(u_i\mid Z_i)=0$ give
\[
\E(u'\Delta_{n,q}u\mid Z)
=
\sum_i\Delta_{n,q,ii}\E(u_i^2\mid Z_i)
=o_p(n),
\]
using the average diagonal convergence and bounded conditional variances.
The conditional variance is bounded by
\[
C\sum_{i\ne j}\Delta_{n,q,ij}^2
+C\sum_i\Delta_{n,q,ii}^2
=o_p(n^2)
\]
from \eqref{eq:matrix-L2} and the conditional fourth-moment bound.
Chebyshev's inequality therefore gives
$n^{-1}u'\Delta_{n,q}u=o_p(1)$ uniformly over $q$, proving
\eqref{eq:kernel-reduction}.

For the bootstrap, conditional on the data,
$u_i^*=\xi_i\widehat u_i$. The Rademacher identities
$\E^*(\xi_i\xi_j)=0$ for $i\ne j$ and $\xi_i^2=1$ produce the same diagonal
mean and off-diagonal variance bounds, now weighted by empirical residual
moments. The sample $(4+\eta)$ moment is $O_p(1)$, so truncation at a level
that diverges slower than $n^{1/(4+\eta)}$ makes the preceding bounds
uniformly integrable; removing the truncation changes the conditional law by
$o_p(1)$. Finally,
$\widehat u_i=u_i-\bar u$ with $\bar u=O_p(n^{-1/2})$. In the multiplier
quadratic form, the pure $\bar u^2$ term is $o_p^*(1)$ after division by $n$
because $\operatorname{tr}(M_{n,q})=O_p(n)$, and the cross term is
$o_p^*(1)$ by Cauchy--Schwarz and \eqref{eq:matrix-row}. This proves the
conditional bootstrap reduction.
\end{proof}

\begin{assumption}[Primitive iid null regularity]
\label{ass:iid-null}
Under $H_0^{CM}$, observations are iid,
$\E[u\mid Z]=0$, and for some constants
$0<\underline\sigma^2<\overline\sigma^2<\infty$ and $\eta>0$,
\[
\underline\sigma^2
\le
\E[u^2\mid Z]
\le
\overline\sigma^2
\quad\text{a.s.},
\qquad
\E[|u|^{4+\eta}\mid Z]\le C
\quad\text{a.s.}
\]
Assumption~\ref{ass:iid-grid} holds. At least one indexed
$K_{e_q}$ is nonzero in $L^2(P_Z\times P_Z)$.
\end{assumption}

For $W=(Z,u)$ define
\begin{equation}
h_q(W,W')
=
K_{e_q}(Z,Z')uu'.
\label{eq:app-hq}
\end{equation}
Under Assumption~\ref{ass:iid-null}, $h_q$ is square-integrable and
canonical:
\[
\E\{h_q(W,W')\mid W\}=0.
\]
It also defines a positive trace-class operator on $L^2(P_W)$. Positivity
follows because, for any $f\in L^2(P_W)$ and
$a_f(z)=\E[u f(W)\mid Z=z]$,
\[
\E\{f(W)h_q(W,W')f(W')\}
=
\E\{a_f(Z)K_{e_q}(Z,Z')a_f(Z')\}\ge0.
\]
Bounded $K_{e_q}$ and $\E u^2<\infty$ imply
\[
\int h_q(w,w)\,dP_W(w)
=
\E\{K_{e_q}(Z,Z)u^2\}<\infty.
\]
For a positive integral operator this diagonal integral equals the operator
trace, so the eigenvalues are nonnegative and absolutely summable.

\begin{theorem}[Primitive iid Rademacher validity for the raw radius vector]
\label{thm:rademacher}
Under Assumptions~\ref{ass:iid-grid} and \ref{ass:iid-null}, let
\[
V_n
=
\left\{
n^{-1}T_n(\widehat e_{n,q})
:q\in\mathcal Q
\right\}
\]
and let $V_n^*$ be the corresponding vector computed from
$Y_i^*=\bar Y+\xi_i\widehat u_i$. Then there exists a nondegenerate finite
vector $\mathcal V$ of second-order Gaussian-chaos variables such that
\[
V_n\Rightarrow\mathcal V
\]
and, conditionally on the data,
\[
d_{\mathrm{BL}}\!\left(
\mathcal L^*(V_n^*),\mathcal L(\mathcal V)
\right)
\xrightarrow{p}0,
\]
where $d_{\mathrm{BL}}$ is bounded-Lipschitz distance.

For each radius, if
$\{\lambda_{q,r}\}_{r\ge1}$ are the nonnegative eigenvalues of the
operator generated by $h_q$, the marginal limit can be represented as
\[
\mathcal V_q
\overset{d}{=}
\sum_{r\ge1}\lambda_{q,r}Z_{q,r}^2,
\]
with the cross-radius dependence induced by the common underlying Gaussian
process.
\end{theorem}

\begin{proof}
Lemma~\ref{lem:ball-matrix} reduces the observed statistic, uniformly over
the finite grid, to
\[
\overline V_{n,q}
=
\frac1n\sum_{i=1}^n\sum_{j=1}^n h_q(W_i,W_j).
\]
The same lemma reduces the bootstrap statistic to the corresponding
multiplier V-statistic with factors $\xi_i\xi_j$.

Because the operator generated by $h_q$ is positive and trace class, choose
an orthonormal eigen-expansion
\[
h_q(w,w')
=
\sum_{r\ge1}\lambda_{q,r}
\phi_{q,r}(w)\phi_{q,r}(w'),
\qquad
\lambda_{q,r}\ge0,
\qquad
\sum_{r\ge1}\lambda_{q,r}<\infty .
\]
Canonical degeneracy implies that every eigenfunction associated with a
nonzero eigenvalue has mean zero. For finite $R$, define
\[
S_{n,q,r}
=
n^{-1/2}\sum_{i=1}^n\phi_{q,r}(W_i).
\]
The rank-$R$ statistic is exactly
\[
\overline V_{n,q}^{(R)}
=
\sum_{r=1}^R\lambda_{q,r}S_{n,q,r}^2.
\]
Collect the finitely many $S_{n,q,r}$ over all
$q\in\mathcal Q$ and $r\le R$. The multivariate central limit theorem gives
joint convergence to a centred Gaussian vector whose covariance between
coordinates $(q,r)$ and $(q',r')$ equals
$\E[\phi_{q,r}(W)\phi_{q',r'}(W)]$. Continuous mapping therefore gives
joint convergence of the rank-$R$ radius-statistic vector.

For the bootstrap, set
\[
S_{n,q,r}^*
=
n^{-1/2}\sum_{i=1}^n\xi_i\phi_{q,r}(W_i).
\]
Conditionally on the sample, its covariance matrix is the empirical
second-moment matrix of the collected eigenfunctions and converges in
probability to the same population covariance matrix. For the finite vector
$\Phi_R(W_i)$ collecting these eigenfunctions, the conditional Lindeberg
term satisfies, for every $\epsilon>0$,
\[
\frac1n\sum_{i=1}^n
\|\Phi_R(W_i)\|^2
\1\{\|\Phi_R(W_i)\|>\epsilon\sqrt n\}
\xrightarrow{p}0,
\]
because $\E\|\Phi_R(W)\|^2<\infty$. The triangular-array Lindeberg--Feller
theorem applied conditionally to
$n^{-1/2}\xi_i\Phi_R(W_i)$ therefore gives the same Gaussian limit for the
finite vector of $S_{n,q,r}^*$. Squaring and summing with the fixed
eigenvalues gives the same rank-$R$ second-order chaos.

It remains to remove the finite-rank truncation. Positivity and trace-class
summability give
\[
\E\left[
\overline V_{n,q}
-
\overline V_{n,q}^{(R)}
\right]
=
\sum_{r>R}\lambda_{q,r},
\]
while canonical degeneracy gives a variance bound proportional to
$\sum_{r>R}\lambda_{q,r}^2$. Both tails vanish as $R\to\infty$. The
conditional multiplier tail obeys the empirical analogue of the same bound;
the law of large numbers for the squared eigenfunctions makes that bound
uniformly tight in probability. Because $\mathcal Q$ is finite, the tail
control is uniform over $q$. First let $n\to\infty$ for fixed $R$, then let
$R\to\infty$.

Finally, Lemma~\ref{lem:ball-matrix} transfers the result from the
population-kernel V-statistics back to the observed Ball matrices,
sample-standardised predictors, empirical distance-quantile radii, and
centred residual multiplier statistic. This establishes both the
unconditional vector limit and conditional bootstrap convergence in
bounded-Lipschitz distance. The finite-rank argument follows the standard
architecture of weighted-bootstrap results for degenerate U- and
V-statistics \citep{arconesgine1992bootstrap}; the quadratic-form
representation provides an alternative classical route
\citep{dejong1987clt}.
\end{proof}

\begin{corollary}[Conditional validity of the implemented iid raw maximum]
\label{cor:raw-max}
Let
\[
T_n^{\max}
=
\max_{q\in\widehat{\mathcal Q}_n}
T_n(\widehat e_{n,q})
\]
and define $T_n^{*,\max}$ analogously. Under
Theorem~\ref{thm:rademacher}, coverage is asymptotically inactive and
\[
n^{-1}T_n^{\max}\Rightarrow
\max_{q\in\mathcal Q}\mathcal V_q.
\]
The limiting maximum has no atoms on $(0,\infty)$. Indeed, every nonzero
coordinate is a nonnegative weighted sum of squared Gaussian variables and
therefore has a continuous distribution; for $x>0$,
\[
\Pp\{\max_q\mathcal V_q=x\}
\le
\sum_{q:K_{e_q}\ne0}\Pp(\mathcal V_q=x)=0.
\]
Because at least one indexed kernel is nonzero, the $(1-\alpha)$ quantile is
positive for $0<\alpha<1$. Hence the ideal conditional Rademacher critical
value gives asymptotic level $\alpha$ without an additional continuity
assumption. A Monte Carlo implementation with $B=B_n\to\infty$ and the
usual plus-one correction inherits this result.
\end{corollary}

\begin{remark}[Scope of the iid result]
Theorem~\ref{thm:rademacher} closes the earlier primitive-condition gap for
the implemented iid raw maximum on a fixed finite radius grid. Its
assumptions remain substantive. Predictor support must be regular and
full-dimensional; local probability mass is uniformly positive at the
indexed radii; the pairwise-distance density must be regular at the indexed
quantiles; the number of radius indices is fixed; and the theorem is
pointwise in the data-generating law. It does not cover singular predictor
manifolds, a grid whose cardinality grows with $n$, or serial observations.
The finite-sample heteroskedastic simulations therefore remain informative,
but they no longer stand in for an unproved iid Ball-matrix transfer.
\end{remark}

\section{Time-Series Calibration}
\label{app:serial}

This appendix records the final time-series calibration and its evidential
status. It replaces the earlier circular-shift theorem. Circular shifts remain
sensitivity checks because they preserve outcome ordering but need not
preserve predictor-linked conditional variance.

\subsection{Algorithm}

Let $\widetilde Y_t=Y_t-\bar Y$. Fit candidate autoregressions
\[
\widetilde Y_t=\sum_{j=1}^{p}\phi_j\widetilde Y_{t-j}+\eta_t,
\qquad p\in\{0,\ldots,6\},
\]
by least squares, excluding candidate fits that fail the implementation's
stability check. Select $\widehat p$ with the prespecified comparable-sample BIC
rule. Retain one-step innovations $\widehat\eta_t$ at their original dates
for $t=\widehat p+1,\ldots,n$. For bootstrap draw $b$, generate iid
Rademacher signs and set
$\eta_t^{*(b)}=\xi_t^{(b)}\widehat\eta_t$. Initialise the recursion with the
observed centred values $\widetilde Y_1,\ldots,\widetilde Y_{\widehat p}$,
recursively generate the remaining observations, and add $\bar Y$. For
$\widehat p=0$, set $Y_t^{*(b)}=\bar Y+\eta_t^{*(b)}$. Bootstrap implementation holds
the predictor path fixed and recomputes all 71 radii, coverage screening, and
the raw maximum.

This method preserves observed innovation magnitudes and dates. Recursive
recolouring restores fitted linear outcome persistence. Unlike a Gaussian
dependent multiplier, it does not impose one covariance bandwidth on every
process. Implementation records report the numerical stability tolerance. The
analysis conditions on the version that passed every recorded stability
check.

\subsection{Formal serial validity class}

The recursive bootstrap is a sign randomisation after a fitted stable linear
filter. A conditional sign-symmetry assumption therefore yields a formal
serial result that matches the implemented algorithm more closely than a
generic weak-dependence multiplier theorem.

Let $P=6$ be the fixed maximum candidate order. For the true order
$p_0\le P$, write
\begin{equation}
u_t=\sum_{j=1}^{p_0}\phi_j u_{t-j}+\eta_t,
\qquad
\eta_t=a_t s_t,
\qquad a_t\ge0.
\label{eq:app-serial-sign}
\end{equation}
Let
\[
\mathcal A_n
=
\sigma\{X_1,\ldots,X_n,
u_1,\ldots,u_{p_0},
a_{p_0+1},\ldots,a_n\}.
\]

\begin{assumption}[Stable recursive-sign class]
\label{ass:serial-sign}
The AR polynomial in \eqref{eq:app-serial-sign} is stable and its parameter
lies in the interior of a compact stable set. Conditional on
$\mathcal A_n$, the signs
$s_{p_0+1},\ldots,s_n$ are independent Rademacher variables. For some
$\delta>0$,
\[
0<c\le n^{-1}\sum_{t=p_0+1}^n a_t^2\le C<\infty
\quad\text{with probability approaching one},
\qquad
n^{-1}\sum_{t=p_0+1}^n a_t^{4+\delta}=O_p(1).
\]
The finite candidate set contains the unique minimal true order $p_0$.
Writing $V_t=(u_{t-1},\ldots,u_{t-p_0})'$, assume
\[
n^{-1}\sum_{t=p_0+1}^n V_tV_t'\;\xrightarrow{p}\;Q_V,
\qquad Q_V\text{ positive definite},
\]
\[
n^{-1/2}\sum_{t=p_0+1}^n V_t\eta_t=O_p(1),
\qquad
n^{-1}\sum_{t=p_0+1}^n\eta_t^2\xrightarrow{p}\sigma_\eta^2>0.
\]
For every underfitted candidate order, the probability limit of its residual
variance exceeds $\sigma_\eta^2$ by a positive constant; for every correctly
specified overfitted candidate, the reduction in residual sum of squares is
$O_p(1)$. Predictor geometry is $\mathcal A_n$-measurable and the fixed-grid
Ball matrices obey \eqref{eq:matrix-entry}--\eqref{eq:matrix-diag}.
\end{assumption}

Assumption~\ref{ass:serial-sign} permits predictor-linked and time-varying
innovation magnitudes. It is stronger than
$\E(\eta_t\mid\mathcal F_{t-1},X_t)=0$: conditioning on the complete
predictor path also rules out feedback from the current innovation sign into
future conditioned predictors. That restriction is the price of an exact
conditional sign-randomisation argument.

\subsection{Oracle recursive-sign randomisation}

For known $(p_0,\phi)$ and fixed observed initial values, let
$b_n(\phi)$ denote the deterministic contribution of those initial values and
let $L_n(\phi)$ map dated innovations to the remaining AR trajectory. Then
\[
u=b_n(\phi)+L_n(\phi)D_a s,
\]
where $D_a=\operatorname{diag}(a_{p_0+1},\ldots,a_n)$ after the obvious
zero-padding of the initial coordinates. Stability gives uniformly bounded
row and column $\ell_1$ norms for $L_n(\phi)$.

\begin{proposition}[Exact oracle recursive-sign validity]
\label{prop:serial-oracle}
Under Assumption~\ref{ass:serial-sign}, conditional on $\mathcal A_n$, the
observed finite-grid vector
\[
\{u'M_{n,q}u:q\in\mathcal Q\}
\]
has exactly the same distribution as the oracle bootstrap vector obtained by
replacing $s$ with an independent Rademacher vector $\xi$ and recursively
recolouring $D_a\xi$ through the true AR filter. The equality in law holds
for the complete raw maximum and for any statistic measurable with respect
to the fixed predictor geometry. A Monte Carlo randomisation test that ranks
the observed statistic with iid oracle sign draws and uses the plus-one
correction has conditional size no greater than $\alpha$.
\end{proposition}

\begin{proof}
Conditional on $\mathcal A_n$, every object in the statistic except the sign
vector is fixed. Assumption~\ref{ass:serial-sign} gives
$s\overset d=\xi$ conditionally on that sigma-field. Applying the same
deterministic map
\[
r\mapsto
\left\{
[b_n(\phi)+L_n(\phi)D_a r]'M_{n,q}
[b_n(\phi)+L_n(\phi)D_a r]
:q\in\mathcal Q
\right\}
\]
to $s$ and $\xi$ proves equality of the complete vectors and hence their
maxima. Conditional rank validity with the plus-one correction is the usual
Monte Carlo randomisation argument.
\end{proof}

\subsection{Stable-filter and estimation replacement}

The feasible procedure replaces $(p_0,\phi,a_t)$ by BIC, least-squares
coefficients, and fitted innovation magnitudes. Three lemmas control those
replacements.

\begin{lemma}[Stable AR filtering preserves Ball leverage]
\label{prop:serial-filter}
Suppose a fixed-radius Ball matrix $M_{n,q}$ satisfies
\eqref{eq:matrix-entry}--\eqref{eq:matrix-diag}. For a stable fixed-order AR
filter define
\[
\widetilde M_{n,q}
=
L_n(\phi)'M_{n,q}L_n(\phi).
\]
Then, uniformly over the fixed radius grid,
\[
\max_{r,s}|\widetilde M_{n,q,rs}|=O_p(1),
\qquad
\max_r\sum_s\widetilde M_{n,q,rs}^2=O_p(n),
\qquad
\sum_r\widetilde M_{n,q,rr}^2=O_p(n).
\]
\end{lemma}

\begin{proof}
Absolute summability of the stable impulse response bounds
$\|L_n(\phi)\|_1$ and $\|L_n(\phi)\|_\infty$ uniformly in $n$. Each entry of
$L_n'M_{n,q}L_n$ is therefore a weighted sum of entries of $M_{n,q}$ with
uniformly bounded total absolute weight, proving the entry bound. The same
convolution argument applied to squared rows gives the row-square bound; its
diagonal restriction gives the final display. Finiteness of the radius set
makes the statement uniform in $q$.
\end{proof}

\begin{lemma}[Comparable-sample BIC and least-squares replacement]
\label{lem:serial-bic}
Under Assumption~\ref{ass:serial-sign}, fit every candidate order
$p\in\{0,\ldots,P\}$ on the common dated subsample
$t=P+1,\ldots,n$. Let
\[
\mathrm{BIC}(p)
=
\log\{\mathrm{RSS}(p)/m\}
+\frac{p\log m}{m},
\qquad m=n-P.
\]
Then $\Pp(\widehat p=p_0)\to1$. On that event,
\[
\|\widehat\phi-\phi\|=O_p(n^{-1/2}),
\qquad
\sum_t(\widehat\eta_t-\eta_t)^2=O_p(1),
\qquad
\max_t|\widehat\eta_t-\eta_t|=o_p(1).
\]
The same rates hold when the fitted regression uses
$Y_t-\bar Y$ rather than $Y_t-\mu$.
\end{lemma}

\begin{proof}
For an underfitted order $p<p_0$, the population linear projection omits at
least one nonzero AR component. Positive-definiteness of the lag covariance
matrix implies a fixed positive gap between the limiting residual variance
and the innovation variance. The BIC penalty is $o(1)$, so every underfitted
order is rejected with probability approaching one.

For $p>p_0$, the additional population coefficients are zero. The reduction
in total RSS from the redundant regressors is $O_p(1)$ under the standard
least-squares martingale expansion, whereas the total BIC penalty difference
is $(p-p_0)\log m\to\infty$. Hence every overfitted order is rejected and
$\widehat p=p_0$ with probability approaching one. This is the standard
fixed-order BIC separation argument; compare the finite-order component of
sieve-bootstrap analyses such as \citet{buhlmann1997sieve}.

On the selected true order,
\[
\widehat\phi-\phi
=
\left(\sum_t U_{t-1}U_{t-1}'\right)^{-1}
\sum_tU_{t-1}\eta_t
=O_p(n^{-1/2}),
\]
where $U_{t-1}$ is the $p_0$-lag vector. Therefore
$\widehat\eta_t-\eta_t=-(\widehat\phi-\phi)'U_{t-1}$ and
\[
\sum_t(\widehat\eta_t-\eta_t)^2
\le
\|\widehat\phi-\phi\|^2\sum_t\|U_{t-1}\|^2
=O_p(1).
\]
The $(4+\delta)$ moment bound implies
$\max_t\|U_{t-1}\|=o_p(n^{1/2})$, yielding the uniform residual rate. Finally,
$\bar Y-\mu=O_p(n^{-1/2})$ under the stable finite-order model. Replacing
$\mu$ by $\bar Y$ adds a common $O_p(n^{-1/2})$ term and its stable lagged
propagation, which preserves the stated $\ell_2$ and uniform rates.
\end{proof}

\begin{lemma}[Feasible bootstrap is asymptotically oracle]
\label{lem:serial-feasible-replacement}
Let $V_n^{o,*}$ be the oracle recursive-sign radius vector from
Proposition~\ref{prop:serial-oracle}, divided by $n$, and let $V_n^{f,*}$ be
the implemented vector using $(\widehat p,\widehat\phi,\widehat\eta_t)$.
Under Assumption~\ref{ass:serial-sign},
\[
\max_{q\in\mathcal Q}
|V_{n,q}^{f,*}-V_{n,q}^{o,*}|
=o_p^*(1)
\]
in outer probability. Consequently
\[
d_{\mathrm{BL}}\{\mathcal L^*(V_n^{f,*}),
\mathcal L^*(V_n^{o,*})\}\xrightarrow{p}0.
\]
\end{lemma}

\begin{proof}
Work on the event $\widehat p=p_0$, whose probability tends to one by
Lemma~\ref{lem:serial-bic}. Stability makes the impulse-response map
continuously differentiable on a neighbourhood of $\phi$, so
\[
\|L_n(\widehat\phi)-L_n(\phi)\|_1
+\|L_n(\widehat\phi)-L_n(\phi)\|_\infty
=O_p(n^{-1/2}).
\]
Let $\widehat a_t=|\widehat\eta_t|$. Lemma~\ref{lem:serial-bic} gives
$\|\widehat a-a\|_2=O_p(1)$ and
$\max_t|\widehat a_t-a_t|=o_p(1)$. The $(4+\delta)$ moment condition also
implies $\max_t a_t=o_p(n^{1/(4+\delta)+\epsilon})$ for every fixed
$\epsilon>0$.

Write the oracle and feasible bootstrap trajectories as
$u^{o,*}=b+UD_\xi\mathbf1$ and
$u^{f,*}=\widehat b+\widehat U D_\xi\mathbf1$, where
$U=L_n(\phi)D_a$, $\widehat U=L_n(\widehat\phi)D_{\widehat a}$, and
$D_\xi$ is the diagonal sign matrix. Expanding the difference of quadratic
forms produces a Rademacher quadratic term, a Rademacher linear term from
the initial-condition contribution, and a deterministic remainder.

For the quadratic term, it is useful to separate filtering from
innovation magnitudes. Write
\[
\widetilde M_{n,q}(\phi)=L_n(\phi)'M_{n,q}L_n(\phi),
\qquad
A_{n,q}=D_a\widetilde M_{n,q}(\phi)D_a,
\]
and define $\widehat A_{n,q}$ analogously. Then
$H_{n,q}=\widehat A_{n,q}-A_{n,q}$ is the matrix in the Rademacher
quadratic-form difference. Decompose
\begin{align*}
H_{n,q}
={}&D_{\widehat a}
\{\widetilde M_{n,q}(\widehat\phi)-\widetilde M_{n,q}(\phi)\}
D_{\widehat a}\\
&+D_{\widehat a-a}\widetilde M_{n,q}(\phi)D_a
+D_a\widetilde M_{n,q}(\phi)D_{\widehat a-a}\\
&+D_{\widehat a-a}\widetilde M_{n,q}(\phi)D_{\widehat a-a}.
\end{align*}
Stability and differentiability of the impulse response imply, uniformly in
$r,s,q$,
\[
|\widetilde M_{n,q,rs}(\widehat\phi)
-\widetilde M_{n,q,rs}(\phi)|
=O_p(n^{-1/2}),
\]
and the corresponding row-square sum is $O_p(1)$. Assumption
\ref{ass:serial-sign} gives
$\max_t a_t=o_p\{n^{1/(4+\delta)+\epsilon}\}$ for every fixed
$\epsilon>0$, while Lemma~\ref{lem:serial-bic} gives
$\|\widehat a-a\|_2=O_p(1)$ and
$\max_t|\widehat a_t-a_t|=o_p(1)$. Lemma~\ref{prop:serial-filter} therefore
yields, for the two single-magnitude-replacement terms,
\begin{align*}
\|D_{\widehat a-a}\widetilde M_{n,q}(\phi)D_a\|_F^2
&\le
\|\widehat a-a\|_2^2
\max_r\sum_s a_s^2\widetilde M_{n,q,rs}(\phi)^2\\
&\le
O_p(1)\,\max_s a_s^2\,O_p(n)
=o_p(n^2).
\end{align*}
The double-replacement term is smaller because
$\max_t|\widehat a_t-a_t|=o_p(1)$. For the filter-parameter term, the
row-square bound for the derivative matrix gives
\[
\|D_{\widehat a}
\{\widetilde M_{n,q}(\widehat\phi)-\widetilde M_{n,q}(\phi)\}
D_{\widehat a}\|_F
\le
O_p(\sqrt n)\max_t\widehat a_t^2
=o_p(n),
\]
because $1/2+2/(4+\delta)<1$. Hence
\[
\|H_{n,q}\|_F=o_p(n)
\]
uniformly over the finite grid.

The trace is controlled separately. Cauchy--Schwarz, the diagonal bound in
Lemma~\ref{prop:serial-filter}, and
$\|\widehat a-a\|_2=O_p(1)$ give $O_p(\sqrt n\max_t a_t)$ for each
single-magnitude replacement. The filter-parameter term is
$O_p(\sqrt n\max_t a_t^2)$ because each diagonal perturbation is
$O_p(n^{-1/2})$. Both are $o_p(n)$ under the $(4+\delta)$ moment condition,
and the double-replacement term is smaller. Thus
\[
|\operatorname{tr}H_{n,q}|=o_p(n).
\]
Conditional Rademacher moments now give
\[
\E^*(\xi'H_{n,q}\xi)=\operatorname{tr}H_{n,q}=o_p(n),
\qquad
\operatorname{Var}^*(\xi'H_{n,q}\xi)
\le2\|H_{n,q}\|_F^2=o_p(n^2).
\]
Hence $n^{-1}\xi'H_{n,q}\xi=o_p^*(1)$.

The difference in initial-condition vectors is $O_p(n^{-1/2})$ in
$\ell_1$ and $\ell_2$ after stable propagation. Combining this with the
Ball row-square bounds shows that the associated Rademacher linear term has
conditional variance $o_p(n^2)$ and its deterministic quadratic term is
$o_p(n)$. Division by $n$ makes both negligible. Finiteness of
$\mathcal Q$ gives the maximum statement. Bounded-Lipschitz convergence
follows because a uniformly Lipschitz test function changes by at most the
sup-norm difference of the two vectors.
\end{proof}

\begin{theorem}[Feasible prewhitened recursive Rademacher validity]
\label{thm:serial-feasible}
Under Assumption~\ref{ass:serial-sign}, let $V_n$ be the observed
71-radius Ball vector divided by $n$, and let $V_n^{f,*}$ be the feasible
bootstrap vector generated by the implemented BIC-selected recursive
Rademacher algorithm. Then
\[
d_{\mathrm{BL}}\left\{
\mathcal L^*(V_n^{f,*}),
\mathcal L(V_n\mid\mathcal A_n)
\right\}
\xrightarrow{p}0.
\]
If the largest conditional atom of the oracle raw maximum converges to zero in probability, the feasible raw-max bootstrap test has asymptotic conditional level $\alpha$ and hence unconditional level $\alpha$.
\end{theorem}

\begin{proof}
Proposition~\ref{prop:serial-oracle} gives exact equality between the oracle
bootstrap law and the conditional law of the observed vector.
Lemma~\ref{lem:serial-feasible-replacement} gives bounded-Lipschitz convergence of the feasible bootstrap law to that oracle law. The triangle inequality proves the first display. The maximum map is Lipschitz on the finite vector. A vanishing maximal conditional atom converts bounded-Lipschitz approximation into convergence of oracle and feasible rejection probabilities at the randomisation cutoff; the oracle plus-one test is conditionally conservative by Proposition~\ref{prop:serial-oracle}. Hence the feasible rejection probability converges to $\alpha$ under asymptotically non-atomic oracle ranks (and remains asymptotically no larger than $\alpha$ under the usual conservative cutoff convention).
\end{proof}

\subsection{Broader martingale-difference scope}

Theorem~\ref{thm:serial-feasible} does not turn conditional sign symmetry
into a consequence of the headline null. A broader working class writes
\[
\mathcal G_t=\sigma(\mathcal F_{t-1},X_t),
\qquad
\E(\eta_t\mid\mathcal G_t)=0,
\]
and imposes the compatibility condition
\[
\sum_{j=1}^{p_0}\phi_j\E(u_{t-j}\mid X_t)=0.
\]
This class allows asymmetric innovations and does not condition on the
future predictor path. It also lies outside the exact sign-randomisation
argument. The application-aligned serial experiments therefore remain
necessary: they test the finite-sample behaviour of the implemented
procedure in persistent and predictor-linked variance designs beyond the
formal sign-symmetric theorem. No claim of model-free validity for arbitrary
stationary, nonlinear, long-memory, or asymmetric martingale-difference
processes is made.

\subsection{Alternative Calibrations}

Fixed-bandwidth Gaussian dependent multipliers fail in opposite directions.
Bandwidth 6 produces mean size 9.41\% and maximum size 47.0\%. Bandwidth 12
reduces mean size to 3.26\% but still reaches 10.8\%. Bandwidths 18 and 24
produce mean sizes 1.79\% and 1.06\%. No fixed bandwidth controls size
stably across the studied process classes.

Circular shifts perform substantially better than iid and short-block
permutations. They nevertheless reach 8.6\% in the original high-persistence
variance-null study and 9.4\% in the auxiliary $n=702,d=4$ design. Shifting breaks the
alignment between the fixed predictor path and the outcome, including the
conditional-variance relation permitted by the null.

\subsection{Simulation Validation}

A corrected $n=240$ confirmation uses eight consequential null cells,
1,000 Monte Carlo replications, and 999 bootstrap draws. Across four
persistent variance-null cells, the selected bootstrap averages 4.48\%
rejection, compared with 8.30\% for circular shifts. Three deliberately
difficult $d=10,\rho=0.9$ cells remain conservative.

Focused power at $n=240,d=10$ retains 95.5\% of aggregate circular-shift
power. Mean difference is $-0.56$ percentage points and the worst
difference is $-3.4$ points. Weak absolute power remains for some local and
disconnected alternatives under short, persistent, high-dimensional designs.

Application-scale validation uses $n=702,d=4$. Across 16 null cells, mean
size is 5.1\%, the range is 4.0\%--6.4\%, and mean absolute distortion is
0.5 percentage points. Application-scale power exceeds circular shifts by
1.07 percentage points on average, with aggregate ratio 1.022. A larger
confirmatory experiment gives marginal sizes of 5.0\% and 5.2\% for four-
and two-axis tests and 4.0\% Holm family-wise error.

These results justify a simulation claim: the selected bootstrap calibrates
the raw maximum well in the process classes and dimensions studied. They do
not establish exact or model-free validity.

\section{Formal Results and Evidence}
\label{app:proof-audit}

\begin{table}[!htbp]
\centering
\caption{Formal result and evidence audit}
\label{tab:proof-audit}
\begin{threeparttable}
\small
\begin{tabularx}{\textwidth}{@{}p{3.0cm}p{2.6cm}Y p{4.0cm}@{}}
\toprule
Result or procedure & Status & Maintained scope & Remaining limitation \\
\midrule
Population null implication & Proved & Integrable scalar outcome & None \\
Local identification & Proved & Compact regular support and a.e. continuity & Excludes singular supports \\
Fixed-radius probability limit & Proved & iid, fixed positive radii, uniform local mass & No growing-grid theory \\
Empirical-quantile corollary & Proved & Fixed quantile indices and regular distance density & No quantile at a mass point \\
Grid-visible fixed-alternative consistency & Proved & Positive population criterion at at least one indexed radius & Not omnibus over all nonconstant means \\
Pitman local-power limit & Proved & $n^{-1/2}$ mean departures projected onto the fixed Ball feature family & No minimax or growing-grid optimality claim \\
Permutation validity & Proved & Conditional exchangeability & Stronger than mean independence \\
Ball-kernel reduction & Proved & iid, fixed finite grid, regular support and distance density & No growing-grid or singular-support theory \\
iid raw Rademacher maximum & Proved & Independent heteroskedastic observations; primitive Ball-kernel conditions & Fixed grid; iid only \\
Oracle recursive-sign calibration & Proved exactly & Stable finite-order AR with conditionally iid innovation signs given predictor path, magnitudes, and initial values & Conditional sign symmetry is stronger than mean independence \\
Feasible prewhitened recursive Rademacher bootstrap & Proved asymptotically & Same sign-symmetric AR class; fixed BIC order set; least-squares replacement & Does not cover general asymmetric MDS, nonlinear dynamics, or long memory \\
Broader serial calibration & Simulation validated & Studied finite-order martingale-difference and predictor-linked variance processes & No general asymmetric-MDS theorem \\
Circular shifts and blocks & Sensitivity evidence & Studied temporal designs & Not primary formal tests of the complete null \\
\bottomrule
\end{tabularx}
\begin{tablenotes}[flushleft]\footnotesize
\item Notes: ``Proved'' refers only to results under the stated maintained
conditions. The iid row concerns the implemented raw finite-grid maximum.
The serial theorem matches the implemented recursive-sign mechanism but is
narrower than the general martingale-difference null class used in the
falsification simulations. Simulation validation therefore remains relevant
for finite-sample performance and for serial processes outside conditional
sign symmetry.
\end{tablenotes}
\end{threeparttable}
\end{table}

\section{Final Implementation Details}
\label{app:implementation}

Each predictor coordinate is standardised using its sample mean and sample standard deviation, and Euclidean distance is computed in that standardised space. Every observed predictor vector is a candidate centre. The primary radius index set is
\[
\mathcal Q=\{.05,.06,\ldots,.75\},
\]
which contains 71 indices. A radius is retained when at least 20\% of centres have at least 10 neighbours. The fixed-radius statistic uses support weights $w_i=N_i$, and the primary statistic is the raw maximum across admissible radii. Predictor geometry is held fixed inside outcome resampling, but coverage and the complete 71-radius maximum are recomputed for every draw.

The .75 endpoint is an empirically validated search limit rather than a uniquely identified structural scale. Endpoint diagnostics for a shorter $q\le .50$ search showed material boundary concentration. A tail sensitivity adds $.80,.85,.90$ but is not the primary search. Selected radius is a descriptive maximiser over this finite domain.

IID calibration uses Rademacher outcome multipliers. Serial finance applications use the BIC-selected prewhitened recursive Rademacher algorithm in Appendix~\ref{app:serial}, with $B=999$. The raw maximum remains primary after the studentisation audit. W1 support weighting remains primary after the capped-support audit. z-score + Euclidean remains primary after MAD and whitening robustness experiments.

\section{Canonical Ball--NCMD Comparison}
\label{app:canonical-extension}

The canonical experiment contains nine DGPs, $n\in\{200,500\}$, $d\in\{2,10,20\}$, two active coordinates, 21 SNR values, and 1,000 replications per cell. NCMD $K=5$ and $K=10$ are evaluated on exact regenerated canonical datasets using the same deterministic seed map as the $q\le .75$ Ball results. All 1,134,000 paired datasets completed without duplicate keys or failed comparisons.

\begin{table}[!htbp]
\centering
\caption{Extended canonical method summary}
\label{tab:method-auc}
\begin{threeparttable}
\small
\begin{tabular}{lrr}
\toprule
Method & Mean full-range AUC & Mean low-SNR AUC \\
\midrule
NCMD $K=10$ & 0.669 & 0.482 \\
$k$NN & 0.654 & 0.474 \\
NCMD $K=5$ & 0.646 & 0.452 \\
Ball $q\le .75$ & 0.637 & 0.465 \\
dCor & 0.523 & 0.370 \\
HSIC & 0.500 & 0.345 \\
KSG & 0.485 & 0.309 \\
HSIC-perm & 0.479 & 0.340 \\
MICe & 0.330 & 0.195 \\
\bottomrule
\end{tabular}
\begin{tablenotes}[flushleft]\footnotesize
\item Notes: Averages span all 54 canonical DGP--sample-size--dimension blocks. The seven established comparators and the two NCMD specifications are evaluated on matched canonical datasets; the NCMD results come from a paired extension using the same seed map. Average AUC should not be read as universal dominance because rankings vary sharply by DGP.
\end{tablenotes}
\end{threeparttable}
\end{table}

Full DGP-level Ball--NCMD AUC values appear in Table~\ref{tab:ncmd-dgp-auc} in the main text.

Canonical pooled null size is 0.05183 for Ball, 0.05169 for NCMD $K=5$, and 0.05133 for NCMD $K=10$. The respective 54-cell ranges are 0.038--0.067, 0.037--0.078, and 0.031--0.071. Both NCMD specifications therefore enter the power comparison.

MDD/MDC remains an important estimand-matched family. The tested public author-code MDD implementation did not advance from the prospective feasibility stage: pooled size was 0.06625, one null-family pooled rate was 0.07625, and the worst cell was 0.105. No post hoc retuning or alternative implementation was substituted.

\section{Predictor-Law and Geometry Robustness}
\label{app:predictor-robustness}

The predictor-robustness experiment fixes $n=500$, $d=10$, two active coordinates, $R=1000$, and SNR $\in\{0,.05,.10,.25,.50\}$. Six predictor laws are used: independent Gaussian; AR(1)-correlated Gaussian with $\Sigma_{jk}=.7^{|j-k|}$; a variance-standardised multivariate $t_5$ with the same correlation; centred and variance-standardised lognormal coordinates; an equal two-component Gaussian mixture shifted by $\pm .75$ in the first two coordinates with within-component variance .4375 there; and $0.98N(0,I)+0.02N(0,25I)$ contamination.

The six signal functions are linear, local island, ring, local absolute signal, thin band, and checkerboard. The internal label \texttt{local\_variance} denotes the local absolute conditional-mean signal, not a variance-only alternative. Each realised signal is centred and its noise variance is recalibrated within replication to the target SNR.

Three predictor geometries are evaluated: sample z-score + Euclidean; median/MAD + Euclidean; and sample-covariance whitening + Euclidean. Whitening is rejected rather than silently regularised if covariance eigenvalues are invalid or ill-conditioned. The mandatory production comparison recorded zero geometry failures.

The compact predictor-law summary appears in Table~\ref{tab:predictor-robustness}.

The table supports a bounded robustness claim. The primary z-score procedure remains approximately calibrated across the six predictor laws. MAD and z-score are close in most laws. Whitening materially reduces power in correlated Gaussian and correlated $t_5$ settings while helping the mixture slightly. No geometry is uniformly optimal.

\section{Finance Robustness and Rolling Evidence}
\label{app:finance-robustness}

\subsection{Data sources and transformations}
\begin{landscape}
\begin{table}[!htbp]
\centering
\caption{Finance data, transformations, and contemporaneous predictor states}
\label{tab:finance-data-sources}
\begin{threeparttable}
\scriptsize
\setlength{\tabcolsep}{3.2pt}
\begin{tabularx}{\linewidth}{@{}p{2.8cm}p{4.4cm}p{4.2cm}X@{}}
\toprule
Outcome & Source series & Transformation and units & Contemporaneous predictors \\
\midrule
\multicolumn{4}{l}{\textit{Panel A: U.S. equity factors}}\\
Market, SMB, HML, RMW, CMA, or Momentum &
Kenneth R. French Data Library: five factors from \nolinkurl{F-F_Research_Data_5_Factors_2x3_CSV.zip}; Momentum from \nolinkurl{F-F_Momentum_Factor_CSV.zip} &
Monthly percentage-point returns; market is \texttt{Mkt-RF} &
The other five factor returns. \\
\addlinespace
\multicolumn{4}{l}{\textit{Panel B: Macro-finance equations}}\\
Market excess return & French \texttt{Mkt-RF} & Monthly percentage points &
\texttt{FEDFUNDS}, inflation, \texttt{UNRATE}, \texttt{T10Y2YM}, \texttt{BAA10YM}, industrial-production growth. \\
Industrial-production growth & FRED \texttt{INDPRO}, seasonally adjusted & $1200\,\Delta\log(\texttt{INDPRO})$, annualised percent &
Market excess return, \texttt{FEDFUNDS}, inflation, \texttt{UNRATE}, \texttt{T10Y2YM}, \texttt{BAA10YM}. \\
Unemployment-rate change & FRED \texttt{UNRATE}, seasonally adjusted & $\Delta\texttt{UNRATE}$, percentage points &
Market excess return, \texttt{FEDFUNDS}, inflation, \texttt{T10Y2YM}, \texttt{BAA10YM}, industrial-production growth. \\
Baa-spread change & FRED \texttt{BAA10YM} & $\Delta\texttt{BAA10YM}$, percentage points; the level is Moody's Baa yield minus the 10-year Treasury yield &
Market excess return, \texttt{FEDFUNDS}, inflation, \texttt{UNRATE}, \texttt{T10Y2YM}, industrial-production growth. \\
\bottomrule
\end{tabularx}
\begin{tablenotes}[flushleft]\footnotesize
\item Notes: Inflation equals $1200\,\Delta\log(\texttt{CPIAUCSL})$, where \texttt{CPIAUCSL} is seasonally adjusted. The federal funds rate is FRED \texttt{FEDFUNDS}; the term spread is FRED \texttt{T10Y2YM}; the Baa spread level is FRED \texttt{BAA10YM}. Public source files were retrieved and archived on 18 July 2026. Complete-case alignment follows all transformations; no observations are imputed. The common sample ends in September 2025, before the documented October 2025 CPI and household-survey gaps, and the factor panel uses the same endpoint for comparability. Predictors are standardised inside each full sample or rolling window. These equations are contemporaneous and do not test forecasting.
\end{tablenotes}
\end{threeparttable}
\end{table}
\end{landscape}

\subsection{Application-aligned serial calibration}
The final raw-versus-studentised serial-null summary appears in Table~\ref{tab:serial-calibration}.

The raw $q\le .75$ procedure is primary. Studentisation slightly raises average canonical power but falls below the prespecified lower serial-size bound and moves selected scales substantially upward.

\subsection{Cross-fitted linear benchmark and residual Ball}
Full-sample original and residual Ball results appear in Table~\ref{tab:full-sample-finance}.
Selected-window original and residual Ball results appear in Table~\ref{tab:selected-diagnostics}.

Five contiguous time-ordered test folds are used. OLS contains an intercept and is estimated in original finance units. HAC inference uses Bartlett/Newey--West weights with lag $\lfloor4(n/100)^{2/9}\rfloor$ and small-sample correction. Residual Ball refits OLS separately inside every bootstrap training fold before computing held-out residuals and the complete $q\le .75$ maximum.

No full-sample residual Ball test rejects. Momentum in July 2022 is the only one of the five selected diagnostic windows that rejects after residualisation. This distinction governs manuscript language about structure beyond ordinary contemporaneous linear comovement.

\subsection{Support weighting}
\begin{table}[!htbp]
\centering
\caption{Support-weighting robustness}
\label{tab:weighting}
\begin{threeparttable}
\small
\begin{tabular}{lr}
\toprule
Diagnostic & W1 support vs. W4 capped support \\
\midrule
Full-sample classification changes & 0 \\
Selected diagnostic classification changes & 0 \\
Holm-family conclusion changes & 0 \\
Median high-union Jaccard & 1.000 \\
Residual Momentum $p$ (W1 / W4) & 0.030 / 0.030 \\
\bottomrule
\end{tabular}
\begin{tablenotes}[flushleft]\footnotesize
\item Notes: W1 uses $w_i=N_i$ and remains primary. W4 caps support weights prospectively. Equal-centre and square-root weights produced serial under-rejection and were not retained. The capped-support alternative leaves the central conclusions unchanged, giving no reason to replace support weighting.
\end{tablenotes}
\end{threeparttable}
\end{table}

Support weighting remains primary. Capping the largest support weights changes no central classification or Holm-family conclusion, and residual Momentum remains at $p=.030$. Equal-centre and square-root alternatives produced serial under-rejection and were not retained.

\subsection{Pointwise rolling $q\le .75$ series}
The primary $q\le .75$ rolling summary appears in Table~\ref{tab:rolling-q075}.

All 3,100 windows completed. Momentum has 56 pointwise rejecting windows and RMW 120; 37 classifications differ from the earlier $q\le .50$ series. The upper boundary $q=.75$ is selected in 122 windows. These summaries are descriptive because adjacent windows overlap by 239 of 240 observations and no simultaneous inference is imposed.

\subsection{Tail search and finance $k$NN calibration}

The $q\le .90$ tail extension leaves unemployment individually significant ($p=.021$) but changes its four-equation Holm value from .028 under $q\le .75$ to .084. This is disclosed as family-wise tail sensitivity. A matched finance $k$NN calibration completed computationally but failed the prespecified serial size criterion, so finance $k$NN p-values are not used for formal inference.

\section{Computational Provenance}
\label{app:repro}

The final procedure is supported by separate, auditable computational components for radius-domain validation, raw-versus-studentised aggregation, predictor-law and metric robustness, linear and cross-fitted residual benchmarks, support weighting, and the paired NCMD comparison. Deterministic seed maps and unique sample identifiers preserve matched comparisons across methods where paired inference is reported.

The NCMD comparison completed 1,134,000 paired canonical datasets without failed or duplicate keys and reproduced the corresponding $q\le .75$ Ball summaries exactly. The broader research archive records configurations, summary outputs, seed manifests, and checksums used to construct the reported exhibits. The manuscript source package is intentionally smaller than the full executable research repository and contains only material needed to compile the working paper.

\end{document}